\documentclass{ejm}

\usepackage{amsmath}
\usepackage[colorlinks,citecolor=blue]{hyperref}
\usepackage{cleveref}
\usepackage{newtxtext}
\usepackage{bm}
\usepackage{cancel}
\usepackage{graphicx}
\usepackage{thmtools}
\usepackage{xcolor}
\hypersetup{colorlinks=true,linkcolor=blue,citecolor=blue,urlcolor=blue}
\usepackage{natbib}
\usepackage{har2nat}
\setcitestyle{comma}
\usepackage{soul}

\let\realverbatim=\verbatim
\let\realendverbatim=\endverbatim
\renewcommand\verbatim{\par\addvspace{6pt plus 2pt minus 1pt}\realverbatim}
\renewcommand\endverbatim{\realendverbatim\addvspace{6pt plus 2pt minus 1pt}}

\ifprodtf \else
  \checkfont{eurm10}
  \iffontfound
    \IfFileExists{upmath.sty}
      {\typeout{^^JFound AMS Euler Roman fonts on the system,
                   using the 'upmath' package.^^J}%
       \usepackage{upmath}}
      {\typeout{^^JFound AMS Euler Roman fonts on the system, but you
                   don't seem to have the}%
       \typeout{'upmath' package installed. EJM.cls can take advantage
                 of these fonts,^^Jif you use 'upmath' package.^^J}%
       \providecommand\mathup[1]{##1}%
       \providecommand\mathbup[1]{##1}%
      }
  \else
    \providecommand\mathup[1]{#1}%
    \providecommand\mathbup[1]{#1}%
  \fi
\fi

\ifprodtf \else
  \checkfont{msam10}
  \iffontfound
    \IfFileExists{amssymb.sty}
      {\typeout{^^JFound AMS Symbol fonts on the system, using the
                'amssymb' package.^^J}%
       \usepackage{amssymb}%
         \let\leq=\leqslant
         \let\geq=\geqslant
      }{}
  \fi
\fi

\ifprodtf \else
  \IfFileExists{amsbsy.sty}
    {\typeout{^^JFound the 'amsbsy' package on the system, using it.^^J}%
     \usepackage{amsbsy}}
    {\providecommand\boldsymbol[1]{\mbox{\boldmath $##1$}}}
\fi

\newcommand\Real{\mbox{Re}} 
\newcommand\Imag{\mbox{Im}} 

\newsavebox{\astrutbox}
\sbox{\astrutbox}{\rule[-5pt]{0pt}{20pt}}

\newcommand\ud{\mathrm{d}}

\newtheorem{proposition}{Proposition}[section]
\newdefinition{remark}{Remark}[section]
\newdefinition{assumption}{Assumption}[section]
\newdefinition{definition}{Definition}[section]
\newdefinition{conjecture}{Conjecture}[section]

\renewcommand{\t}[1]{\text{#1}}
\newcommand{\lt}[1]{\mathrm{LT}\left[#1\right]}

\title[European Journal of Applied Mathematics]
{Fluctuation Response Patterns of~Network~Dynamics -- an Introduction}
\author[X. Zhang et al.]{%
  Xiaozhu\ns Z\ls H\ls A\ls N\ls G$\,^{1,2,3}$\ns
  \and
  Marc\ns T\ls I\ls M\ls M\ls E$\,^{3,4}$\ns
}

\affiliation{%
  $^1$MOE Key Laboratory of Advanced Micro-Structured Materials and School of Physics Science and Engineering, Tongji University, Shanghai 200092, P. R. China\\
  $^2$Frontiers Science Center for Intelligent Autonomous Systems, Tongji University, Shanghai 200092, P. R. China\\
  $^3$Chair for Network Dynamics, Center for Advancing Electronics Dresden (cfaed) and Institute for Theoretical Physics, Technical University of Dresden, 01062 Dresden, Germany\\
  $^4$Lakeside Labs, Lakeside B04b, 9020 Klagenfurt, Austria\\
    email\textup{\nocorr: \texttt{xiaozhu.zhang@tongji.edu.cn;marc.timme@tu-dresden.de}}}

\date{\today}
\pubyear{2021}
\volume{000}
\pagerange{\pageref{firstpage}--\pageref{lastpage}}

\begin{document}

\label{firstpage}
\maketitle

\begin{abstract}
Networked dynamical systems, i.e., systems of dynamical units coupled via nontrivial interaction topologies, constitute models of broad classes of complex systems, ranging from gene regulatory and metabolic circuits in our cells to pandemics spreading across continents. Most of such systems are driven by irregular and distributed fluctuating input signals from the environment. Yet how networked dynamical systems collectively respond to such fluctuations depends on the location and type of driving signal, the interaction topology and several other factors and remains largely unknown to date.
As a key example, modern electric power grids are undergoing a rapid and systematic transformation towards more sustainable systems, signified by high penetrations of renewable energy sources. These in turn introduce significant fluctuations in power input and thereby pose immediate challenges to the stable operation of power grid systems. How power grid systems dynamically respond to fluctuating power feed-in as well as other temporal changes is critical for ensuring a reliable operation of power grids yet not well understood.
In this work, we systematically introduce a linear response theory for fluctuation-driven networked dynamical systems. The derivations presented not only provide approximate analytical descriptions of the dynamical responses of networks, but more importantly, allows to extract key qualitative features about spatio-temporally distributed response patterns.   Specifically, we provide a general formulation of a linear response theory for perturbed networked dynamical systems, explicate how dynamic network response patterns arise from the solution of the linearized response dynamics, and emphasize the role of linear response theory in predicting and comprehending power grid responses on different temporal and spatial scales and to various types of disturbances.
Understanding such patterns from a general, mathematical perspective enables to estimate network responses quickly and intuitively, and to develop guiding principles for, e.g., power grid operation, control and design.
\end{abstract}


\begin{subjclass}[2020]
Primary: 34C15, 94C15; Secondary: 93A14, 34E10
\end{subjclass}

\section{Introduction}

Networked dynamical systems abound around and in us. From the circuits supporting metabolism and gene regulation in our cells to the neural networks in our brains, from the power grids that supply electric energy to most of our technical infrastructure to the internet that connects our computers, all of these systems are driven externally, often by irregular, time-dependent and spatially heterogeneous signals. How networked dynamical systems respond to such perturbations, driving signals, or other types of time-dependent inputs is hardly understood to date. In this article, we offer a general introduction to the basic theory of analyzing response features of networked dynamical systems by linear response theory and focus on applications in the realm of power grids.

A reliable supply with electric power fundamentally underlies most aspects of modern society. As the shares of renewable electric energy supply grow and consumer dynamics is increasingly influenced by digital technologies, fluctuations impinge on power grids, making them intrinsically driven, non-equilibrium systems, with distributed and often non-stationary response dynamics \citet{witthaut_collective_2022}. If fluctuations become large, they may destabilize grid dynamics, affect normal operations of parts of the grid, cause cascades of failures or even total blackouts [\citet{investigation_committee_final_2007,wilde_ofgem_2020,entso-e_continental_2021}]. To predict, control or mitigate the fluctuating and distributed responses of such networks resulting from input (and output) power fluctuations, we need to understand their nature in network dynamical systems. 
    
How can we systematically characterize fluctuating responses that are distributed across meshed networks? How can we predict their influences and at which nodes in a network is their impact most profound? Linear response theory (LRT) relates sufficiently small time-dependent driving signals to time-dependent responses. The theory approximates a resulting system dynamics near some operating point to first order in the strength of the driving signals. It has traditionally found applications across physics, chemistry and engineering [\citet{kubo_statistical-mechanical_1957,cammi_linear_1999,ikeguchi_protein_2005,ruelle_review_2009,majda_linear_2010,pan_non-hermitian_2020}], and recently also in power grid models [\citet{dorfler_synchronization_2013,witthaut_critical_2016,kettemann_delocalization_2016,manik_network_2017,tamrakar_propagation_2018,haehne_propagation_2019,tyloo_key_2019}]. However, a framework of linear response theory uncovering spatiotemporal response patterns in systems with multi-dimensional dynamics of units that simultaneously interact via intricate network topologies is not yet well established.

In this article, we introduce a general formalism of linear response theory for networked dynamical systems and demonstrate its applications in stationary and non-stationary models of power grids. Specifically, in Sec.~\ref{sec:general}, we present the main ideas of the linear response theory for networks by first analyzing the linear responses of networked dynamical systems with the most general settings (Sec.~\ref{subsec:general_formulation}) and then boiling down step-by-step to a specific LRT which provides a direct link between the temporal and the spatial features of the dynamic network responses (Sec.~\ref{subsec:general_operators}). In Sec.~\ref{sec:LRT_models}, we demonstrate the LRT for networks by applying it to two models of power grids: a stationary model for the DC approximation of the AC power flow distributions in power grids (Sec.~\ref{subsec:LRT_models_DC}), and a non-stationary model, the (second-order) oscillator model, commonly used to describe the dynamics of the high-voltage AC power transmission networks (Sec.~\ref{subsec:LRT_oscillator_model}). The next two sections, Sec.~\ref{sec:patterns} and Sec.~\ref{sec:role_LRT}, focus on the extraction and the interpretation of spatiotemporal response patterns of power grids from the LRT, as well as the analytical techniques needed therein. In Sec.~\ref{sec:patterns}, we elaborate how distinctive steady-state response patterns emerge in three frequency regimes (Sec.~\ref{subsec:steady_patterns}), how the transient spreading pattern of a perturbation entangles with the underlying network topology (Sec.~\ref{subsec:transient_patterns}), and how LRT helps to estimate the long-term risk of individual nodes from external fluctuations (Sec.~\ref{subsec:dvi}). In Sec.~\ref{sec:role_LRT} we summarize the role of LRT in uncovering such patterns by comparing the analyses for patterns in the steady-state response vs. the transient responses (Sec.~\ref{subsec:role_steady_transient}), the patterns in the deterministic responses to a given perturbation vs. the cumulative responses to a random signal (Sec.~\ref{subsec:role_deterministic_stochastic}), and the small responses close to the base operation state vs. the large responses further away (Sec.~\ref{subsec:role_small_large}). In the last section (Sec.~\ref{sec:outlook}), we point out several directions for extending the present theory for a better understanding of the response dynamics of networked dynamical systems and for a safer and more reliable operation of future power grids.

\section{Main ideas of the linear response theory for networks}
\label{sec:general}
In this section we introduce the main ideas of the linear response theory framework. Sec.~\ref{subsec:general_formulation} formulates the LRT on a general model of networks of one-dimensional dynamical units. Sec.~\ref{subsec:general_operators} highlights the linear operators, represented as matrices, arising in the LRT for specific network settings and how basic information of the underlying network system such as its topology and features defining the dynamically determined operating point enter those operators. In Sec.~\ref{subsec:general_operators} we also derive the explicit linear responses of the general network model of one-dimensional units introduced in Sec.~\ref{subsec:general_formulation} and extend the results for networks of higher-order nodal dynamics.

\subsection{General formulation of LRT for networked dynamical systems}
\label{subsec:general_formulation}

We now illustrate the main idea underlying LRT by starting with a general setting of networked dynamical systems where each unit's dynamics is one-dimensional. We consider a dynamical process involving $N$ interacting variables, whose state is represented by a vector $\bm{x}=(x_1,\cdots,x_N)\in \mathbb{R}^N$. The autonomous dynamics of the $N$-dimensional dynamical system is governed by
\begin{equation}
    \dot{\bm{x}}=\bm{f}(\bm{x}),
    \label{eq:dynamics_general}
\end{equation}
where $\bm{f}:\mathbb{R}^N\rightarrow\mathbb{R}^N$ is a function that in general depends on the states $\bm{x}$ of all units and does not explicitly depend on time. Let us consider a system that exhibits a fixed point at $\bm{x}=\bm{x}^*$ where $\bm{f}(\bm{x}^*)=\bm{0}$, and is influenced by an external dynamic driving vector $\bm{D}(t)\in\mathbb{R}^N$ at the fixed point $\bm{x}^*$. We investigate how the system dynamically responds to $\bm{D}(t)$, i.e. how the system's deviation $\bm{X}(t):=\bm{x}(t)-\bm{x}^*$ from the fixed point  evolves in time. In general, the dynamics of the system's response $\bm{X}$ follows
\begin{equation}
    \dot{\bm{X}}=\bm{f}(\bm{x}^*+\bm{X})+\bm{D}(t),
    \label{eq:response_dynamics_general}
\end{equation}
where both functions $\bm{f}$ and $\bm{D}$ can be highly nonlinear so that an exact analytical solution of the system's response $\bm{X}$ typically does not exist or is unknown.

Important information about the response dynamics (\ref{eq:response_dynamics_general}) is given by  the stability operator at the fixed point obtained from the linearization of the function $\bm{f}$ at $\bm{x}=\bm{x}^*$, i.e. the Jacobian matrix $\mathcal{J}$ with $\mathcal{J}_{ij}:=\frac{\partial f_i}{\partial x_j}|_{\bm{x}=\bm{x}^*}$. The system's response dynamics thus approximately follows the linearized differential equation
\begin{equation}
    \dot{\bm{X}}=\mathcal{J}\bm{X}+\bm{D}(t).
    \label{eq:LRT_general}
\end{equation}
The signs of the real parts of the Jacobian eigenvalues $\bm{w}^{[\ell]}$ (with $\ell$ being the index) indicate whether the fixed point $\bm{x}^*$ is linearly stable ($\Real[\bm{w}^{[\ell]}]<0$), unstable ($\Real[\bm{w}^{[\ell]}]>0$) or neutrally stable ($\Real[\bm{w}^{[\ell]}]=0$) in the eigendirection/eigenspace corresponding to the eigenvalue $\bm{w}^{[\ell]}$. Below, we typically focus on the dynamics near stable fixed points, where all $\Real[\bm{w}^{[\ell]}]<0$ (or marginally stable ones where one $\Real[\bm{w}^{[\ell]}]=0$ if $\mathcal{J}$ is a Laplacian operator) to justify the assumption that the solution $\bm{X}(t)$ of the linearized equation \eqref{eq:LRT_general} reasonably well approximates the full solution of \eqref{eq:response_dynamics_general}.

Let us now focus on systems with \textit{pairwise} interactions between the $N$ units such that the function $f_i(\bm{x})$ controling the time evolution of unit $i$ can be written as
\begin{equation}
    f_i(\bm{x})=h_i(x_i)+\sum_{j=1; \ j\neq i}^N K_{ij}g_{ij}(x_i,x_j),
    \label{eq:pairwise_interaction}
\end{equation}
where $h_i$ denotes the intrinsic dynamics depending on the variable $x_i$ itself and the coupling term $K_{ij}g_{ij}(x_i,x_j)$ represents the pairwise interaction between variable $x_i$ and $x_j$ with $i\neq j$. Here $K_{ij}\in\mathbb{R}^+_0$ is the non-negative coupling strength and $g_{ij}$ is the coupling function depending on the state of the pair of variables $(x_i, x_j)$ and $g_{ij}(x_i,x_j)\not \equiv 0$. 

\begin{remark}
Pairwise interactions (\ref{eq:pairwise_interaction}) induce an interaction structure of the dynamical system (\ref{eq:dynamics_general}) that can be represented by a \textit{graph} $G(V,E)$ where the set of vertices $V$ consists of $N$ variables $x_1,\cdots, x_N$ and the set of edges $E$ consists of all node pairs with the pairwise coupling strength being nonzero, i.e. $E=\{(i,j)\in V^2 | \  (i\neq j) \land (K_{ij}\neq0)\}$. 
\end{remark}

The linearization of the pairwise interaction (\ref{eq:pairwise_interaction}) yields the Jacobian matrix of the response dynamics (\ref{eq:response_dynamics_general}) of the driven networked system
\begin{subeqnarray}
    \mathcal{J}_{ij} &=& \frac{\partial}{\partial x_j}\left.\left({h}_i({x}_i)+\sum_{k\neq i} K_{ik}{g}_{ik}\left({x}_i,{x}_k\right)\right)\right|_{\bm{x}=\bm{x}^*}\\[3mm]
	&=& \left\{\begin{array}{ll}
	\left.\dfrac{\ud {h}_i}{\ud {x}_i}\right|_{\bm{x}=\bm{x}^*}+\displaystyle\sum_{k\neq i} K_{ik}\left.\dfrac{\partial{g}_{ik}}{\partial {x}_i}\right|_{\bm{x}=\bm{x}^*} &\mathrm{for}\quad j=i \\
	K_{ij}\left.\dfrac{\partial{g}_{ij}}{\partial {x}_j}\right|_{\bm{x}=\bm{x}^*} &\mathrm{for}\quad j\neq i.
	\end{array}\right.
	\label{eq:Jacobian_pairwise}
\end{subeqnarray}
It is clear that the topology of the interaction network $G$ explicitly enters the Jacobian matrix through the matrix of the coupling strength $K\in\mathbb{R}^{N\times N}$ with its $ij$-th element being defined as $K_{ij}$, which is effectively a weighted adjacency matrix of graph $G$. 

The combination of the linearized dynamics of a general $N$-dimensional networked dynamical system (\ref{eq:dynamics_general}) near a fixed point (\ref{eq:LRT_general}) and the arising Jacobian matrix that explicitly depends on the system's topology (\ref{eq:Jacobian_pairwise}) due to pairwise interactions (\ref{eq:pairwise_interaction}) provides a general form of LRT for networked dynamical systems. For a specific system with given forms of or constraints for the intrinsic nodal dynamics $h_i(x_i)$, coupling strengths $K_{ij}$ and coupling function $g_{ij}(x_i,x_j)$, the solution of the matrix equation (\ref{eq:LRT_general}), $X(t)$, can be characterized by evaluating the specific spectral properties of the Jacobian matrix $\mathcal{J}$. Thereby the LRT provides a powerful tool to describe the dynamic response of a networked system temporally and spatially at once: the solution $X(t)$ explicitly depends on time, and at the same time its relation to the topology-dependent Jacobian matrix can be exploited to reveal the temporally and spatially distributed patterns of the network response.

\subsection{Arising linear operators and network topology}
\label{subsec:general_operators}

As discussed in the last section, the dependence of the Jacobian matrix of network response dynamics on a weighted adjacency matrix of the underlying interaction network gives the first hint on how temporal and spatial features of dynamic network responses are intertwined. In this section we further show that, under a few commonly satisfied conditions such as diffusive coupling between units, interesting results of network response dynamics emerge. Especially, another important graph-theoretical matrix, the Laplacian matrix, arises in the network response dynamics (\ref{eq:LRT_general}), providing us powerful tools for characterizing the spatiotemporal patterns in dynamic network responses.

Diffusive coupling is a very common type of coupling present in many physiological and chemical systems [\citet{hale_diffusive_1997,larter_coupled_1999,postnov_neural_2006,casagrande_synchronization_2006,stankovski_coupling_2017}], in particular also appearing in the Kuramoto model [\citet{kuramoto_chemical_1984}] and its variations  [\citet{filatrella_analysis_2008,acebron_kuramoto_2005}]. A diffusive coupling function $\tilde{g}_{ij}$ mediating the interaction between a pair of nodes $(i,j)$ is characterized by its dependence on the state difference $x_j-x_i$ of the node pair, i.e. $g_{ij}(x_i,x_j)=\tilde{g}_{ij}(x_j-x_i)$ in \eqref{eq:Jacobian_pairwise}. For notational simplicity, we again denote the functions $\tilde{g}_{ij}$ just by $g_{ij}$.

\begin{proposition}[Stability of diffusively coupled networks: a special case]\label{prop:stability}
A networked dynamical system with evolution function (\ref{eq:pairwise_interaction}) and interaction network $G(V,E)$ is at least neutrally stable at its fixed point $\bm{x}=\bm{x}^*$ if a) the intrinsic nodal dynamics $h_i$ satisfies $\frac{\ud h_i}{\ud x_i}|_{\bm{x}=\bm{x}^*}\leq0$ for all nodes $i\in V$, b) the coupling function $g_{ij}$ is diffusive and satisfies $\frac{\ud g_{ij}}{\ud(x_j-x_i)}|_{\bm{x}=\bm{x}^*}\geq0$ for all node pairs $(i,j)\in E$.
\end{proposition}

\begin{proof}
The diffusive form of the coupling function $g_{ij}(x_j-x_i)$ yields a useful relation $\frac{\partial g_{ij}}{\partial x_j}=-\frac{\partial g_{ij}}{\partial x_i}=\frac{\ud g_{ij}}{\ud(x_j-x_i)}$. With this particular symmetry, the Jacobian matrix (\ref{eq:Jacobian_pairwise}) of the system at the fixed point $\bm{x}=\bm{x}^*$ takes the following form
\begin{equation}
    \mathcal{J}_{ij}=\left\{\begin{array}{ll}
	-\beta_i-\displaystyle\sum_{k\neq i} K_{ik}\gamma_{ik} &\mathrm{for}\quad i=j \\
	K_{ij}\gamma_{ij} &\mathrm{for}\quad i\neq j,
	\end{array}\right.
	\label{eq:Jacobian_diffusive}
\end{equation}
where $\beta_i:=-\left.\frac{\ud h_i}{\ud x_i}\right|_{\bm{x}=\bm{x}^*}$ and $\gamma_{ij}:=\left.\frac{\ud g_{ij}}{\ud(x_j-x_i)}\right|_{\bm{x}=\bm{x}^*}$. Given that $\beta_i\geq0$, $\gamma_{ij}\geq0$ and $K_{ij}\geq0$ by definition, the Jacobian $\mathcal{J}$ is diagonally dominant:
\begin{equation}
    \left|\mathcal{J}_{ii}\right|=\left|-\beta_i-\displaystyle\sum_{k\neq i} K_{ik}\gamma_{ik}\right|
	=\beta_i+\displaystyle\sum_{k\neq i} K_{ik}\gamma_{ik}
	\geq\displaystyle\sum_{k\neq i} K_{ik}\gamma_{ik}
	=\displaystyle\sum_{k\neq i} \left|K_{ik}\gamma_{ik}\right|
	=\displaystyle\sum_{j\neq i} \left|\mathcal{J}_{ij}\right|.
	\label{eq:Jacobian_diagonally_dominant}
\end{equation}
According to the Gershgorin circle theorem, every eigenvalue of the Jacobian matrix lies within at least one of the $N$ Gershgorin discs $D_i(\mathcal{J}_{ii},r_i)=\left\lbrace z\in\mathbb{C} |\, \,\left|z-\mathcal{J}_{ii}\right|\leq r_i \right\rbrace$ with $r_i=\sum_{k\neq i}|\mathcal{J}_{ki}|$ in the complex plane. Relation (\ref{eq:Jacobian_diagonally_dominant}) indicates that all $N$ Gershgorin discs lie in the left half of the complex plane, i.e. in $\left\lbrace  z\in \mathbb{C}\left|\,\Real\left(z\right)\leq 0\right.\right\rbrace$, because the centers of the discs $(\mathcal{J}_{ii},0)$ lie on the negative real axis and the radius of the discs $r_i=\sum_{k\neq i}|\mathcal{J}_{ki}|\leq|\mathcal{J}_{ii}|$. The discs touch the imaginary axis from the half plane $\left\lbrace  z\in \mathbb{C}\left|\,\Real\left(z\right)\leq 0\right.\right\rbrace$ only if $\beta_i=0$. Therefore, all Jacobian eigenvalues can only have non-positive real parts and consequently the networked dynamical system is at least neutrally stable at the fixed point $\bm{x}=\bm{x}^*$.
\end{proof}

\begin{remark}[Homogeneous nodal dynamics and graph Laplacian]
\label{remark:homogeneous_dynamics_Laplacian}
We remark that in case of identical intrinsic nodal dynamics at all nodes, a weighted Laplacian matrix $\mathcal{L}$ of the interaction graph explicitly enters the linearized response dynamics of the network (\ref{eq:LRT_general}). Assuming $\beta_i=\beta\in\mathbb{R}$ for all nodes $i$, we can express the Jacobian matrix as
\begin{equation}
    \mathcal{J}=- \beta \bm{1} - \mathcal{L},
    \label{eq:Jacobian_identical_nodal_dynamics}
\end{equation}
where the weighted graph Laplacian $\mathcal{L}$ is defined as
\begin{equation}
    \mathcal{L}_{ij}:=\left\{\begin{array}{ll}
	\displaystyle\sum_{k\neq i} K_{ik}\gamma_{ik} &\mathrm{for}\quad i=j \\
	-K_{ij}\gamma_{ij} &\mathrm{for}\quad i\neq j.
	\end{array}\right.
	\label{eq:Laplacian_general}
\end{equation}
Here $K_{ij}\gamma_{ij}$ is considered a weight of edge $(i,j)$, containing the coupling strength $K_{ij}$ and the linearized coupling function $\gamma_{ij}$ at the fixed point. 
\end{remark}

\begin{remark}[Symmetry and linear responses in Laplacian eigenbasis]
\label{remark:symmetry_response_Laplacian}
In general, the interaction network can be directed, meaning that for an edge $(i,j)$ the coupling strength $K_{ij}$ and the derivative of the coupling function $\gamma_{ij}$ at the fixed point, i.e. the sensitivity of the diffusive coupling function $g_{ij}(x_j-x_i)$ to a change in the state difference $x_j-x_i$, can differ from their counterparts $K_{ji}$ and $\gamma_{ji}$ for the edge $(j,i)$ with the opposite direction. This asymmetry leads to an asymmetric weighted graph Laplacian (\ref{eq:Laplacian_general}). However, for undirected networks with symmetric strengths ($K_{ij}=K_{ji}$) and symmetric sensitivities of coupling functions ($\gamma_{ij}=\gamma_{ji}$), or more generally, a symmetric combination $K_{ij}\gamma_{ij}=K_{ji}\gamma_{ji}$, the weighted graph Laplacian $\mathcal{L}$ is symmetric. Its eigenvectors thus form an orthogonal basis which allows us to solve for the linear network responses in (\ref{eq:LRT_general}) and (\ref{eq:Jacobian_identical_nodal_dynamics}) by expressing the response vector $\bm{X}(t)$ in terms of the eigenvalues and eigenvectors of the Laplacian.

We first analyze a system that is perturbed by only one sinusoidal signal with magnitude $\varepsilon>0$ and frequency $\omega>0$ at node $k$, i.e. $D^{(k)}_i(t)=\delta_{ik}\varepsilon e^{\imath(\omega t+\varphi)}$, where $\delta_{ik}$ is the Kronecker delta with $\delta_{ik}=1$ if $i=k$ and $\delta_{ik}=0$ otherwise. We solve for the linear network response vector $\bm{X}^{(k)}(\omega,t)$ governed by
\begin{equation}
    \dot{\bm{X}}^{(k)}=-\beta\bm{X}^{(k)}-\mathcal{L}\bm{X}^{(k)}+\bm{D}^{(k)}(t).
    \label{eq:response_dynamics_Laplacian}
\end{equation}
Expressing the response vector in the constant eigenbasis of Laplacian \begin{equation}
    \bm{X}^{(k)}(t)=\sum_{\ell=0}^{N-1}c^{[\ell]}(t)\bm{v}^{[\ell]}
    \label{eq:linear_response_Laplacian_basis}
\end{equation} 
and the time-dependent coefficients $c^{[\ell]}(t)$ and exploiting orthogonality of the Laplacian eigenvectors, we obtain the ordinary differential equations 
\begin{equation}
    \dot{c}^{[\ell]}=(-\beta-\lambda^{[\ell]})c^{[\ell]}+\varepsilon v^{[\ell]}_k e^{\imath(\omega t+\varphi)}
    \label{eq:coefficient_ODE}
\end{equation}
for the coefficients.
Here the $N$ Laplacian eigenvalues $\lambda^{[\ell]}$ and eigenvectors $\bm{v}^{[\ell]}$ are indexed according to $0=\lambda^{[0]}\leq\lambda^{[1]}\leq\cdots\leq\lambda^{[N-1]}$. If the graph is connected, the zero eigenvalue is unique, such that all other eigenvalues are positive real numbers [\citet{bapat_graphs_2010}]. The solution of differential equation \eqref{eq:coefficient_ODE} for the coefficients is 
\begin{equation}
    c^{[\ell]}(\omega,t)=\left(\bm{X}_0^{(k)}\cdot\bm{v}^{[\ell]}\right)e^{\left(-\beta-\lambda^{[\ell]}\right)t}+\dfrac{\varepsilon v^{[\ell]}_k e^{\imath\varphi}}{\beta+\lambda^{[\ell]}+\imath\omega}\left(-e^{\left(-\beta-\lambda^{[\ell]}\right)t}+e^{\imath\omega t}\right),
    \label{eq:coefficient_general}
\end{equation}
where $\bm{X}_0^{(k)}$ denotes the initial response vector at $t=0$ given node $k$ is perturbed. The linear network response is thus given by (\ref{eq:linear_response_Laplacian_basis}) and (\ref{eq:coefficient_general}).
\end{remark}

\begin{remark}[LRT for distributed arbitrary perturbations]
\label{remark:distributed_arbitrary_perturbations}
In case perturbations with arbitrary temporal structures are distributed across the network, that is, multiple elements of the perturbation vector $\bm{D}(t)$ are arbitrary time series, the linear network response can be obtained by summing up the single-signal single-frequency response (\ref{eq:linear_response_Laplacian_basis} and \ref{eq:coefficient_general}) over all frequency components $\omega$ and all perturbation signals $k$ by resorting to the linearity of the dynamics (\ref{eq:LRT_general})
\begin{equation}
    \bm{X}(t)=\sum_{k}\left(\bm{X}^{(k)}_{\omega=0}(t)+\sum_{\omega>0}\bm{X}^{(k)}(\omega,t)\,\ud\omega\right).
\end{equation}
Here $\bm{X}^{(k)}_{\omega=0}(t)$ denotes the linear response to a constant ($\omega=0$) perturbation $\varepsilon$, which shares the same form as $\bm{X}^{(k)}(\omega,t)$ with $\omega=0$ for $\beta>0$ but not for $\beta=0$. Moreover, the sum becomes an integral for continuously distributed frequencies, with $\bm{X}^{(k)}(\omega,t)$ replaced by the associated response density per unit frequency.

For systems where nodal damping vanishes ($\beta=0$), the coefficient for the $0$-th eigenmode $c^{[0]}$ diverges for a constant perturbation as $\omega \rightarrow 0$, so it needs to be solved separately and takes the form of $\bm{X}_0^{(k)}\cdot\bm{v}^{[0]}+\varepsilon v_k^{[0]} t$, inducing a linear drift and thus unbounded growth with time. As the dynamics LRT describes is only approximate for nonlinear dynamical system, unbounded growth typically induces the approximation of the linear response to the real system to break down, i.e. become useless in practice due to larger errors between approximate and exact solution.
\end{remark}

\begin{remark}[LRT for higher-order nodal dynamics]
    In the above paragraphs we discussed the main ideas of the LRT for networked dynamical systems with first-order nodal dynamics. For more general systems with second- or higher-order nodal dynamics, the straightforward relation (\ref{eq:Jacobian_identical_nodal_dynamics}) between the Jacobian matrix and a weighted graph Laplacian does not hold any more. Nevertheless, a symmetric weighted graph Laplacian still arises in the linearized response dynamics for diffusively-coupled undirected networks with symmetric coupling strengths and symmetric sensitivities of coupling functions as discussed in \autoref{remark:symmetry_response_Laplacian}. If the higher-order time derivatives of the state variables has homogeneous coefficients for individual nodes, i.e. the response dynamics to a perturbation $\bm{D}(t)$ has the form of 
    \begin{equation}
        \sum_d\kappa_dD_t^{d}\bm{X}=-\mathcal{L}\bm{X}+\bm{D}(t),
    \end{equation}
    an explicit solution of the linear responses in the eigenbasis of $\mathcal{L}$ can still be obtained following the routine in \autoref{remark:symmetry_response_Laplacian}, if the corresponding ODEs for the time-dependent coefficients
    \begin{equation}
        \sum_d\kappa_dD_t^{d}{c}^{[\ell]}=-\lambda^{[\ell]}c^{[\ell]}+\bm{D}(t)\cdot\bm{v}^{[\ell]}
    \end{equation}
    are solvable. Here $\kappa_d\in\mathbb{R}$ are constant coefficients and we use Euler's notation for derivatives, where $D_t^{d}x$ denotes the $d$-th time derivative of variable $x$.

\end{remark}

In summary, if a networked dynamical system consisting of $N$ units
\begin{enumerate}
    \item has a fixed point at $\bm{x}=\bm{x}^*$,
    \item in the neighbourhood of $\bm{x}=\bm{x}^*$ the intrinsic nodal dynamics $h_i(x_i)$ gives homogeneous non-positive feedback to the respective nodes, i.e. $\beta_i:=-\left.\frac{\ud h_i}{\ud x_i}\right|_{\bm{x}=\bm{x}^*}=\beta\geq0$ for all $i$, and
    \item the coupling function $g_{ij}(x_j-x_i)$ is diffusive and the coupling term's sensitivity to small changes at $\bm{x}=\bm{x}^*$ is symmetric and non-negative, i.e. $K_{ij}=K_{ji}$ and $\gamma_{ij}:=\left.\frac{\ud g_{ij}}{\ud(x_j-x_i)}\right|_{\bm{x}=\bm{x}^*}=\gamma_{ji}\geq0$ for all edges $(i,j)$,
\end{enumerate}
then i) the dynamical system is at least neutrally stable (\autoref{prop:stability}) and ii) a symmetric weighted graph Laplacian (\ref{eq:Laplacian_general}) arises in the network response dynamics (\ref{eq:response_dynamics_Laplacian}) and enables the expression of linear network responses in the Laplacian eigenbasis (\ref{eq:linear_response_Laplacian_basis}) and (\ref{eq:coefficient_general}) (\autoref{remark:homogeneous_dynamics_Laplacian} and \autoref{remark:symmetry_response_Laplacian}). As we will show in Section \ref{sec:patterns}, the explicit dependence of the linear network response on Laplacian eigenvalues and eigenvectors provides a powerful tool to reveal how the dynamic responses spatially distributed across the network.

\section{LRT for power grid models}
\label{sec:LRT_models}
We now discuss how the LRT for general networked dynamical systems introduced in Section \ref{sec:general} applies to stationary and non-stationary models of power grids and helps reveal static and dynamic responses of power grid systems to external perturbations.

\subsection{LRT for the DC power flow model}
\label{subsec:LRT_models_DC}

We first demonstrate how LRT works in a minimal model, the DC power flow model, and how it helps to compute the systemic stationary response of a power transmission network to perturbations in power injections and withdrawals. 

For common AC power grids, the full power flow analysis poses several challenges such as possible difficulties in finding a solution in ill-conditioned cases and the existence of multiple solutions due to the inherent nonlinearities [\citet{milano_power_2010}]. By linearizing the AC power flow equations at an operation point, the DC power flow model\footnote{The DC power flow model here must not be confused with models for high-voltage direct-current (HVDC) transmission grids, which actually uses DC, as opposed to alternating current (AC), for the transmission of electrical power.} provides a relatively simple and computational inexpensive way to compute the power flows. 

In an AC power transmission grid, the total complex power flow $S_j$ from unit $j$ to unit $i$ reads
\begin{equation}
    S_{ij}=U_{j} I^*_{ij}
    =U_j\left(\dfrac{U_j-U_i}{Z_{ij}}\right)^*
    =\dfrac{|U_j|e^{\imath\theta_j}(|U_j|e^{-\imath\theta_j}-|U_i|e^{-\imath\theta_i})}{R_{ij}-\imath X_{ij}},
    \label{eq:Sij_definition}
\end{equation}
where $U_j=|U_j|e^{\imath\theta_j}$ and $I_{ij}=(U_j-U_i)/Z_{ij}$ denote the voltage at node $j$ and the current between nodes $j$ and $i$, respectively. Both are expressed as complex numbers to reflect the oscillating nature of AC power generation. The asterisks in, e.g., $I_{ij}^*$ indicate the complex conjugates (e.g., of $I_{ij}$). Moreover,  $Z_{ij}=R_{ij}+\imath X_{ij}$ denotes the impedance of the transmission line $(i,j)$ between unit $i$ and $j$, with $R_{ij}$ 
denoting the resistance and $X_{ij}$ the reactance of $(i,j)$. The complex power flow $S_{ij}=P_{ij}+\imath Q_{ij}$ consists of two parts, the active power $P_{ij}$ that results in net energy transfer and the reactive power $Q_{ij}$ that returns to the source in each cycle, not doing any work, but supporting the voltage stability of the power system [\citet{machowski_power_2008}].

The DC power flow model is based on the following assumptions on the parameters and the operating state of the power grid systems:
\begin{assumption}[Perfect voltage support]
The voltage amplitude is constant and identical for each node in the power grid network, $|U_i|\equiv U$ for all $i$, and the management of the reactive power is not considered.
\label{assump:perfect_voltage}
\end{assumption}
\begin{assumption}[Lossless lines]
Transmission losses on the lines are neglected, implying that line resistances are negligible compared to line reactances: $R_{ij}/ X_{ij}\rightarrow 0$ for all lines $(i,j)$.
\label{assump:lossless_lines}
\end{assumption}
\begin{assumption}[Low line loads]
Loads on all transmission lines are low, that is, the voltage angle differences between all neighboring nodes are much smaller in magnitude than $\pi/2$ such  that $\sin(\theta_j-\theta_i)\approx\theta_j-\theta_i$ and $\cos(\theta_j-\theta_i)\approx1$.
\label{assump:low_loads}
\end{assumption}

With the above mentioned assumptions in mind, the complex power flow (\ref{eq:Sij_definition}) between neighboring nodes simplifies to
\begin{equation}
    S_{ij}=\dfrac{U^2(1-e^{\imath(\theta_j-\theta_i)})}{-\imath X_{ij}}
    =\dfrac{U^2}{X_{ij}}\left(\theta_j-\theta_i\right) = P_{ij}.
    \label{eq:Sij_DCapproximation}
\end{equation}
Here the complex power flow $S_{ij}$ naturally reduces to the active power flow $P_{ij}$ since the imaginary part vanishes. The equation (\ref{eq:Sij_DCapproximation}) resembles the expression of the direct current carried by a ``resistor'' $X_{ij}/U^2$ influenced by a ``voltage drop'' $\theta_j-\theta_i$ according to Ohm's law, hence the name ``DC power flow model''.

For an AC power grid system consisting of $N$ units, the active power flow $P_i$ injected at unit $i$ is the sum 
\begin{equation}
    P_i=\sum_{j=1}^{N} P_{ij}=\sum_{j=1}^N K_{ij}\left(\theta_i-\theta_j\right)
    \label{eq:electrical_power_transmission}
\end{equation}
over all connected units. Equation \eqref{eq:electrical_power_transmission} is the core of the DC power flow model as it yields the pattern of power flows $K_{ij}\theta_j$ across the grid network.
Here we follow the notation introduced in Sec.~ \ref{subsec:general_formulation} and define the coupling strength as $K_{ij}=U^2/X_{ij}$ if there exists a transmission line between unit $i$ and $j$ and $K_{ij}=0$ if there is not. Denoting the nodal active power injections and nodal voltage angles as vectors, $\bm{P}:=(P_1,\cdots,P_N)$ and  $\bm{\theta}:=(\theta_1,\cdots,\theta_N)$, we express the linear relation between them by a weighted graph Laplacian $\mathcal{L}$ introduced in Section \ref{subsec:general_operators}, i.e.
\begin{equation}
    \bm{P}=\mathcal{L}\bm{\theta},
    \label{eq:DC_matrix_equation}
\end{equation}
Here $\mathcal{L}$ is defined similarly as in (\ref{eq:Laplacian_general}), only with $\gamma_{ij}\equiv 1$ for all edges $(i,j)$. 

Assuming that the power transmission network runs at a normal operation state where the voltage angles $\bm{\theta}^*$ are stationary at all nodes, the fixed voltage angle differences determine a specific power flow pattern  $\bm{P}^*$ across the network through the linear operator $\mathcal{L}$ such that $\bm{P}^*=\mathcal{L}\bm{\theta}^*$. If the nodal power injections are perturbed as $\bm{P}(t)=\bm{P}^*+\bm{D}(t)$ by a shift vector $\bm{D}$(t) that in general is time dependent and reflects an increase or decrease 
of power generation for consumption at some of the nodes, the nodal voltage angles $\bm{\theta}=\bm{\theta}^*+\bm{\Theta}$ change accordingly through $\mathcal{L}$ due to the linear operation in (\ref{eq:DC_matrix_equation}). The response vector of the voltage angles $\bm{\Theta}$ is given by
\begin{equation}
    \bm{\Theta}(t)=\mathcal{L}^{+}\bm{D}(t),
    \label{eq:LRT_DC_power_flow}
\end{equation}
where $\mathcal{L}^{+}$ denotes the Moore-Penrose inverse of the weighted graph Laplacian $\mathcal{L}$.

\begin{remark}
The weighted graph Laplacian $\mathcal{L}$ defined in (\ref{eq:Laplacian_general}) is singular since it always has an eigenvalue $\lambda_{0}=0$ with the corresponding eigenvector $\bm{v}_0=(1,\cdots,1)$ satisfying $\mathcal{L}\bm{v}_0=\bm{0}$ by construction. Therefore, to compute the voltage angle shifts (\ref{eq:LRT_DC_power_flow}) we need the generalized inverse matrix $\mathcal{L}^{+}$, which can be computed by, e.g., the singular value decomposition. Alternatively, we can remove one dimension of the system by treating the phase $\theta_k$ of one of the units $k$ as the reference for voltage angle, i.e. by setting  $\theta_i\rightarrow\theta_i-\theta_k$ for all $i$. If the network is connected, then the $(N-1)$-dimensional Laplacian matrix is invertible and both of the matrix equations, (\ref{eq:DC_matrix_equation}) and (\ref{eq:LRT_DC_power_flow}), have a unique solution respectively.
\end{remark}

\subsection{LRT for the oscillator model of AC power grid dynamics}
\label{subsec:LRT_oscillator_model}

In this section we discuss how LRT applies for the oscillator model of AC power grids, a widely used model for analysing the dynamics of AC power grids, and thereby provides a way to accurately determine the high-dimensional dynamic responses of an arbitrary power grid network to fluctuating power injections and withdrawals.

The dynamics of the high-voltage AC power transmission networks is essentially captured by an \textit{oscillator model}  (or second order model) of AC power grids, of which synchronization in terms of networked dynamical systems have been initially studied in references \citeasnoun{filatrella_analysis_2008, rohden_self-organized_2012} and \citeasnoun{motter_spontaneous_2013}. This model allows for analytical understanding of the dynamics of power grids and has yielded fruitful research results over the past decade [\citet{rohden_self-organized_2012,motter_spontaneous_2013,dorfler_synchronization_2013, witthaut_critical_2016,tyloo_key_2019}]. As the name suggests, in the oscillator model, each unit of AC power grids, a synchronous machine, is represented by an oscillator and the power transmission lines are represented by the pairwise couplings between the oscillators. The normal operation state of a power grid corresponds to the synchronization of all oscillators, where all units rotate at the same frequency $\Omega_0^{m}$ corresponding to the nominal grid frequency $\Omega_0=2\pi\times50$ or $2\pi\times60$ Hz. 

For each unit in the oscillator model, a synchronous machine, any change of the angular velocity of rotation results from the imbalance of the torques acting on the rotor operated at the nominal grid frequency. Its dynamics is governed by the so-called swing equation [\citet{kundur_power_1994,machowski_power_2008}]:
\begin{equation}
    I\ddot{\theta}^{\t{m}}+D^{\t{m}}\dot{\theta}^{\t{m}}=T^{\t{m}}-T^{\t{e}},
	\label{eq:swing_equation_original}
\end{equation}
where $\theta^{\t{m}}$ denotes the mechanical rotor angle deviation from the rotating reference frame $\Omega_0 t$, $I$ denotes the moment of inertia of the rotor and the connected turbine, $D^{\t{m}}$ denotes the coefficient of the damping torque resulting from the velocity-dependent friction at the air gap between the rotor and the stator in the synchronous machine. $T^{\t{m}}$ and $T^{\t{e}}$ denote respectively the net mechanical torque and the counteracting electromagnetic torque acting on the rotor. 

For deviations of the angular velocity $\dot{\theta}^{\t{m}}$, the local frequency deviation is small compared to the nominal grid frequency $\Omega_0^{\t{m}}$, i.e. $(\dot{\theta}^{\text{m}}+\Omega_0^{\text{m}})^{-1}\approx(\Omega_0^{\text{m}})^{-1}$, so that a torque $T$ acting on the rotor can be expressed in terms of the power $P$ of the synchronous machine as $T=(\dot{\theta}^{\text{m}}+\Omega_0^{\text{m}})^{-1}P\approx(\Omega_0^{\text{m}})^{-1}P$. \st{and using the relations} Also introducing the effective quantities $\theta(t)=\theta^{\t{m}}(t)/(p/2)$ and $\Omega_0=\Omega^{\t{m}}_0/(p/2)$ between the mechanical quantities and their electrical counterparts for a synchronous machine with $p$ poles per phase, we obtain the more common version of the swing equation describing the dynamical relation between the rate of change of the electrical load angle and the power transmission between units:
\begin{equation}
    M\ddot{\theta}+\tilde{D}\dot{\theta}=P^{\text{m}}-P^{\text{e}}.
    \label{eq:swing_equation}
\end{equation}
Here $M:=I\Omega_0/(p/2)^2$ and $\tilde{D}:=D^{\t{m}}\Omega_0/(p/2)^2$ are respectively the angular momentum of the rotor operated at the nominal grid frequency and the damping coefficient of the machine. On the right hand side of (\ref{eq:swing_equation}), $P^{\text{m}}$ denotes the net injected mechanical power (positive when the machine is operated as a generator and negative when operated as a motor), and $P^{\text{e}}$ denotes the electrical power injected to the grid by the synchronous machine. In the oscillator model we again assume perfect voltage support (\autoref{assump:perfect_voltage}) and lossless transmission lines (\autoref{assump:lossless_lines}), which yield $P^{\text{e}}=\sum_{j=1}^NK_{ij}(\theta_j-\theta_i)$ [cf. $P_i$ in the DC power flow model (\ref{eq:electrical_power_transmission})]. Putting everything together, we obtain the governing equations of the oscillator model of AC power grids
\begin{equation}
    \ddot{\theta}_i=P_i-\alpha_i\dot{\theta}_i+\sum_{j=1}^NK_{ij}\sin(\theta_j-\theta_i),\quad \text{for } i\in\{1,\cdots,N\}
    \label{eq:oscillator_model}
\end{equation}
with the parameters $P_i:=P_i^{\t{m}}/M_i$, $\alpha_i:=\tilde{D}_i/M_i$, and $K_{ij}:=U^2/X_{ij}$. 

\begin{proposition}[Linear stability of the oscillator model]
    AC power grid systems described by the oscillator model (\ref{eq:oscillator_model}) with underlying interaction topology $G(V,E)$ is at least neutrally stable at a fixed point $\bm{\theta}^*$, if all edges are not overloaded, i.e. $|\theta_j^*-\theta_i^*|\leq\frac{\pi}{2}$ for all $(i,j)\in E$.
\end{proposition}
\begin{proof}
    At a fixed point of the system $\bm{\theta}=\bm{\theta}^*$, a small deviation of the oscillators' angles $\bm{\Theta}:=\bm{\theta}-\bm{\theta}^*$ follows the linear dynamics
    \begin{align}
        \dfrac{\mathrm{d}}{\mathrm{d}t}\begin{pmatrix} \bm{\Theta} \\ \dot{\bm{\Theta}} \end{pmatrix}
	    =\mathcal{J}\begin{pmatrix} \bm{\Theta} \\ \dot{\bm{\Theta}} \end{pmatrix},
	    \label{eq:thetaFirstOrderSystem}
    \end{align}
    where the Jacobian matrix $\mathcal{J}\in\mathbb{R}^{2N\times2N}$ of the $2N$-dimensional dynamical system is given by
    \begin{align}
	\mathcal{J}=
	\begin{pmatrix}
		\bm{0}_N     &  \bm{I}_N\\
		-\mathcal{L} &  -\mathcal{A}.
	\end{pmatrix}
    \end{align}
    Here $\mathcal{L}$ is a weighted graph Laplacian as defined in (\ref{eq:Laplacian_general}) with $\gamma_{ij}=\cos(\theta^*_j-\theta^*_i)$, $\mathcal{A}$ is an $N\times N$ diagonal matrix with $\mathcal{A}_{ii}:=\alpha_i$, and $\bm{0}_N$ and $\bm{I}_N$ are respectively the $N\times N$ zero matrix and identity matrix.
    
    Let $\bm{w}=(\bm{w}_1,\bm{w}_2)\in\mathbb{C}^{2N}$ with $\bm{w}_1,\bm{w}_2\in\mathbb{C}^{N}$ be an eigenvector of  $\mathcal{J}$ corresponding to eigenvalue $\mu\in\mathbb{C}$. By definition we have $\mathcal{J}\bm{w}=\mu\bm{w}$, which by writing \eqref{eq:thetaFirstOrderSystem} as a second order differential equation implies
    \begin{equation}
        \mu^2\bm{w}_1+\mu\mathcal{A}\bm{w}_1+\mathcal{L}\bm{w}_1=\bm{0}.
        \label{eq:oscillator_model_Jacobian_eigenvalue}
    \end{equation}
    Multiplying both sides of the equation above with the conjugate transpose, $\bm{w}_1^{\dagger}$ of $\bm{w}_1$ from the left, we obtain an expression of the eigenvalue $\mu=\left(-\chi_2\pm\sqrt{\chi_2^2-4\chi_1\chi_3}\right)/2\chi_1$, with $\chi_1=\bm{w}^{\dagger}_1\bm{w_1}\geq0$, $\chi_2=\bm{w}^{\dagger}_1\mathcal{A}\bm{w}_1\geq0$ and $\chi_3=\bm{w}^{\dagger}_1\mathcal{L}\bm{w}_1\geq0$. $\chi_2$ and $\chi_3$ are non-negative since $\mathcal{A}_{ii}=\alpha_i>0$ and $\mathcal{L}$ is positive semi-definite because $\gamma_{ij}\geq0$ is ensured through $|\theta_j^*-\theta_i^*|\leq\frac{\pi}{2}$ for all $(i,j)\in E$ (see \autoref{prop:stability} and \autoref{remark:homogeneous_dynamics_Laplacian}) [cf. \citet{manik_supply_2014}]. Therefore the eigenvalue always has a non-positive real part, implying that the networked dynamical system is at least neutrally stable at the fixed point.
\end{proof}

\begin{remark}[Neutral stability and global phase shift]
  Connected AC power grids described by the oscillator model (\ref{eq:oscillator_model}) is neutrally stable at a fixed point only when the deviation $\bm{\Theta}$ is a global phase shift. Because (\ref{eq:oscillator_model_Jacobian_eigenvalue}) indicates that the Jacobian eigenvalue $\mu=0$ only when $\chi_3=0$, which implies $\mathcal{L}\bm{w}_1=\bm{0}$ and thus $\bm{w}_1$ lies in the Laplacian eigenspace corresponding to the eigenvalue $\lambda^{[0]}=0$. For connected graphs there is only one eigenvector $\bm{v}^{[0]}=(1,\cdots,1)$ corresponding to $\lambda^{[0]}=0$, therefore the system is only neutrally stable when all nodes undergo a phase shift with the same magnitude. A global phase shift has no effects on the power flow pattern in the network since pairwise phase differences across edges remain the same.
\end{remark}

\begin{proposition}[LRT of the oscillator model and homogeneous nodal damping]
    Consider an AC power grid oscillator model with arbitrary topology (\ref{eq:oscillator_model}) with homogeneous nodal damping $\alpha_i=\alpha\geq0$ for all nodes $i$. Then the network-wide linear responses to arbitrary external perturbations near a normal operation state $\bm{\theta}^*$ can be expressed explicitly in the eigenbasis of a weighted graph Laplacian: i) The network response to time-independent distributed perturbations $\bm{D}^*$ is
    \begin{equation}
        \bm{\Theta}(t)=\bm{D}^*\cdot\bm{v}^{[0]}\left(\dfrac{e^{-\alpha t}}{\alpha^2}-\dfrac{1}{\alpha^2}+\dfrac{t}{\alpha}\right)\bm{v}^{[0]}+\sum_{\ell=1}^{N-1}\dfrac{\bm{D}^*\cdot\bm{v}^{[\ell]}}{\lambda^{[\ell]}}\left(\dfrac{\Delta_-^{[\ell]}e^{\Delta_+^{[\ell]}t}-\Delta_+^{[\ell]}e^{\Delta_-^{[\ell]}t}}{2\eta^{[\ell]}}+1\right)\bm{v}^{[\ell]},
        \label{eq:LRT_oscillator_step}
    \end{equation}
    and ii) the network response to a single sinusoidal perturbation given by $\bm{D}(t)$ with $D_{i}(t)=\delta_{ik}\varepsilon e^{\imath(\omega t+\varphi)}$ is
    \begin{equation}
        \bm{\Theta}^{(k)}(t)=\sum_{\ell=0}^{N-1}\dfrac{\varepsilon v_k^{[\ell]}e^{\imath\varphi}}{-\omega^2+\imath\alpha\omega+\lambda^{[\ell]}}\left[
        \dfrac{\left(\Delta^{[\ell]}_--\imath\omega\right)e^{\Delta^{[\ell]}_+t}-\left(\Delta^{[\ell]}_+-\imath\omega\right)e^{\Delta^{[\ell]}_-t}}{2\eta^{[\ell]}}+e^{\imath\omega t}\right]\bm{v}^{[\ell]}
        \label{eq:LRT_oscillator_sinusoidal}
    \end{equation}
    with $\Delta^{[\ell]}_{\pm}:=-\alpha/2\pm\eta^{[\ell]}$ and $\eta^{[\ell]}:=\sqrt{\alpha^2/4-\lambda^{[\ell]}}$.
\end{proposition}
\begin{proof}
    Vectorizing the linear response $\bm{\Theta}(t)$ of the system (\ref{eq:oscillator_model}) to a perturbation vector $\bm{D}(t)$, we obtain the matrix equation describing the response dynamics of the oscillator model of AC power grids
    \begin{equation}
        \ddot{\bm{\Theta}}+\bm{\alpha}\circ\dot{\bm{\Theta}}=-\mathcal{L}\bm{\Theta}+\bm{D}(t)
	\label{eq:response_dynamics_oscillator_general},
    \end{equation}
    where $\bm{\alpha}:=(\alpha_1,\cdots,\alpha_N)$ denotes the vector of damping parameters and $``\circ"$ denotes the Schur (element-wise) product of two vectors. Let $\alpha_i=\alpha$ for all $i$, the term $\bm{\alpha}\circ\dot{\bm{\Theta}}$ reduces to a scalar multiplication $\alpha\dot{\bm{\Theta}}$, thereby all terms involving the variable $\bm{\Theta}$ in equation (\ref{eq:response_dynamics_oscillator_general}) can be expressed as linear combinations of Laplacian eigenvectors. Because the $\mathcal{L}$ here is real and symmetric so that we can write $\bm{\Theta}(t)=\sum_{\ell=0}^{N-1}c^{[\ell]}(t)\bm{v}^{[\ell]}$. Using the same trick as in \autoref{remark:symmetry_response_Laplacian}, we obtain equations for the time-dependent coefficients $c^{[\ell]}(t)$ given by 
    \begin{equation}
    	\ddot{c}^{[\ell]}+\alpha\dot{c}^{[\ell]}+\lambda^{[\ell]} c^{[\ell]}=\bm{D}(t)\cdot\bm{v}^{[\ell]}\quad \text{for }\ell\in\{0,\cdots,N-1\}.
    	\label{eq:coefficients_oscillator_model}
    \end{equation}
    Assuming at $t=0$ the AC power grid system operates normally at the fixed point $\bm{\theta}^*$, we have initial conditions for the coefficients $c^{[\ell]}(0)=0$ and $\dot{c}^{[\ell]}(0)=0$ for all $\ell$, thereby obtain explicitly solutions for the coefficients and for the network linear responses. 
    
    For perturbations independent of time,  $\bm{D}(t)=\bm{D}^*$, such as constant shifts in power injection and consumption, the linear response of a power grid system is given by directly solving for the coefficients in (\ref{eq:coefficients_oscillator_model}):
    \begin{equation}
        \bm{\Theta}(t)=\bm{D}^*\cdot\bm{v}^{[0]}\left(\dfrac{e^{-\alpha t}}{\alpha^2}-\dfrac{1}{\alpha^2}+\dfrac{t}{\alpha}\right)\bm{v}^{[0]}+\sum_{\ell=1}^{N-1}\dfrac{\bm{D}^*\cdot\bm{v}^{[\ell]}}{\lambda^{[\ell]}}\left(\dfrac{\Delta_-^{[\ell]}e^{\Delta_+^{[\ell]}t}-\Delta_+^{[\ell]}e^{\Delta_-^{[\ell]}t}}{2\eta^{[\ell]}}+1\right)\bm{v}^{[\ell]},
    \end{equation}
    with $\Delta^{[\ell]}_{\pm}:=-\alpha/2\pm\eta^{[\ell]}$ and $\eta^{[\ell]}:=\sqrt{\alpha^2/4-\lambda^{[\ell]}}$. For time-dependent perturbations $\bm{D}(t)$, such as fluctuating power injections from renewables, we obtain the network linear response based on the responses to each single frequency components at each perturbed nodes, as discussed in \autoref{remark:distributed_arbitrary_perturbations}. Similarly, we let $D_{i}(t)=\delta_{ik}\varepsilon e^{\imath(\omega t+\varphi)}$ and obtain the oscillator model's linear response to a sinusoidal signal at node $k$ as
    \begin{equation}
        \bm{\Theta}^{(k)}(t)=\sum_{\ell=0}^{N-1}\dfrac{\varepsilon v_k^{[\ell]}e^{\imath\varphi}}{-\omega^2+\imath\alpha\omega+\lambda^{[\ell]}}\left[
        \dfrac{\left(\Delta^{[\ell]}_--\imath\omega\right)e^{\Delta^{[\ell]}_+t}-\left(\Delta^{[\ell]}_+-\imath\omega\right)e^{\Delta^{[\ell]}_-t}}{2\eta^{[\ell]}}+e^{\imath\omega t}\right]\bm{v}^{[\ell]}.
    \end{equation}
\end{proof}

\begin{remark}[Low dissipation regime and grid eigenfrequencies]
\label{remark:oscillator_eigenfrequencies}
    In case the dissipation in the system is low enough such that $\alpha<2\sqrt{\lambda^{[\ell]}}$ for the $\ell$-th eigenvalue, in the solution of the linear response (\ref{eq:LRT_oscillator_sinusoidal}) the corresponding $\eta^{[\ell]}$ for the same eigenmode becomes imaginary, suggesting this mode is oscillating under-damped in the power grid system with an exponentially decaying amplitude proportional to $e^{-\frac{\alpha}{2}t}$, i.e., with a time constant $\tau=2/\alpha$. Such intrinsic oscillation modes can also be identified by looking at the eigenvalues of the Jacobian matrix at the fixed point. Let $\bm{w}_1$ in (\ref{eq:oscillator_model_Jacobian_eigenvalue}) be the $\ell$-th eigenvector of the Laplacian $\mathcal{L}$, we can see the corresponding Jacobian eigenvalue $\mu^{[\ell]}=-\alpha/2\pm\sqrt{\alpha^2/4-\lambda^{[\ell]}}=\Delta^{[\ell]}_{\pm}$, which indicates the corresponding eigenfrequency $\omega^{[\ell]}_{\text{eigen}}:=\Imag[\mu^{[\ell]}]=\sqrt{\lambda^{[\ell]}-\alpha^2/4}$ of the power grid system. Since for a connected networked system with $N$ nodes has $N-1$ positive Laplacian eigenvalues, it also has a band of $N-1$ eigenfrequencies if the dissipation is sufficiently low satisfying $\alpha<2\sqrt{\lambda^{[\ell]}}$ for all $N-1$ Laplacian eigenvalues.
\end{remark}

\section{Emerging network response patterns from LRT}
\label{sec:patterns}

In Sec.~\ref{sec:LRT_models} we applied LRT on power grid models and obtained explicit solutions for linear network responses to perturbations at a normal operation state. The solutions are expressed in terms of the eigensystem of a weighted graph Laplacian. These Laplacians and thus their eigensystems contain information about the underlying network topology as well as the base operating state of the system, enabling us to understand and to manipulate how complex networked systems such as power grids collectively respond to external perturbation signals.

In this section we focus on the dynamic responses of AC power grids to time-varying perturbations based on the oscillator model (see Sec.~\ref{subsec:LRT_oscillator_model}) and explicate how steady-state response patterns constituted by the set of nodal response magnitudes as well as transient response patterns describing the spatio-temporal spreading of a perturbation in power grids are mathematically extracted from the solution given by the LRT. 

\subsection{Frequency regimes of steady-state response patterns}
\label{subsec:steady_patterns}

After a transient phase characterized by a dissipation-related time constant $\tau=2/\alpha$ (cf. \autoref{remark:oscillator_eigenfrequencies}), the perturbed power grid systems reside in a second regime of network responses, where the entire network respond periodically to perturbation signals for $t\gg\tau$ (see \ref{eq:LRT_oscillator_step} and \ref{eq:LRT_oscillator_sinusoidal}). We thus call the network responses for such large times \textit{steady-state responses}. We remark that the steady-state responses here are not necessarily stationary, meaning that the nodal responses themselves can vary with time, but their characteristics, such as the amplitude and the phase of sinusoidal responses, constitute network-wide response patterns that are time-independent.

\begin{proposition}[Steady-state response pattern for a constant perturbation]
\label{prop:osillator_pattern_constant}
    For AC power grids with arbitrary topologies (\ref{eq:oscillator_model}), the steady-state responses to time-independent perturbations $\bm{D}(t)=\bm{D}^*$, i.e. with a perturbation frequency $\omega=0$, near a normal operation state $\bm{\theta}^*$ are constituted by a homogeneous shift of grid frequency 
    \begin{equation}
        \delta\dot{\theta}_i=\dfrac{1}{N\alpha}\displaystyle\sum_{j=1}^N D^*_j\quad\t{for }i\in\{1,\cdots,N\},
        \label{eq:pattern1_oscillator_step}
    \end{equation} 
    and a topology-dependent phase shift
    \begin{equation}
        \delta\bm{\theta}=-\dfrac{1}{N\alpha^2}\displaystyle\sum_{i=1}^N D^*_i\bm{1}+\sum_{\ell=1}^{N-1}\dfrac{\bm{D}^*\cdot\bm{v}^{[\ell]}}{\lambda^{[\ell]}}\bm{v}^{[\ell]}.
        \label{eq:pattern2_oscillator_step}
    \end{equation}
\end{proposition}
\begin{proof}
By definition, the steady-state response patterns become clear by investigating the asymptotic behaviour of the responses as $t\rightarrow\infty$. For responses to a constant perturbation vector $\bm{D}(t)=\bm{D}^*$ (\ref{eq:LRT_oscillator_step}) the steady-state response reads
\begin{equation}
    \bm{\Theta}(t)\overset{t\rightarrow\infty}{\sim}
    \bm{D}^*\cdot\bm{v}^{[0]}\left(-\dfrac{1}{\alpha^2}+\dfrac{t}{\alpha}\right)\bm{v}^{[0]}+\sum_{\ell=1}^{N-1}\dfrac{\bm{D}^*\cdot\bm{v}^{[\ell]}}{\lambda^{[\ell]}}\bm{v}^{[\ell]},
\end{equation}
which consists of two characteristic patterns: i) the phases of all units drift away from the normal operation state with a constant angular velocity $(\bm{D}^*\cdot\bm{v}^{[0]})v^{[0]}_i/\alpha$, and ii) a time-independent and unit-specific phase shift. Since $\bm{v}^{[0]}=\frac{1}{\sqrt{N}}\bm{1}$, the former pattern represents the constant homogeneous shift of grid frequency (\ref{eq:pattern1_oscillator_step}) and the latter the topology-dependent phase shift (\ref{eq:pattern2_oscillator_step}).
\end{proof}

\begin{remark}
    The global grid frequency shift (\ref{eq:pattern1_oscillator_step}) is a consequence of the imbalance between power injected into and drawn from the power grid system, which is imposed by the constant perturbation $\bm{D}^*$. The topology-dependent phase shifts (\ref{eq:pattern2_oscillator_step}) at all units suggest a network-wide redistribution of power flows on the transmission lines. The first-order approximation of the change of the load $L_{ij}:=\sin(\theta_j-\theta_i)$ on line $(i,j)$, $\delta L_{ij}=\cos(\theta_j^*-\theta_i^*)(\delta\theta_j-\delta\theta_i)$, provides an indicator for the emerging risks such as overheating for heavily-loaded lines with $L_{ij}$ approaching the upper limit of its safety range.
\end{remark}

\begin{proposition}[Steady-state response pattern for a sinusoidal perturbation]
\label{prop:osillator_pattern_sinusoidal}
    For AC power grids with arbitrary topologies (\ref{eq:oscillator_model}), the steady-state responses to a sinusoidal perturbation at node $k$, i.e. $\bm{D}(t)$ with $D_i(t)=\delta_{ik}\varepsilon e^{\imath(\omega t+\varphi)}$, near a normal operation state $\bm{\theta}^*$ are constituted by a homogeneous phase shift
    \begin{equation}
        \delta\theta_i=\dfrac{\imath\varepsilon e^{\imath\varphi}}{\alpha\omega N}\quad\t{for }i\in\{1,\cdots,N\},
    \end{equation}
    and sinusoidal responses at each node with the same frequency $\omega$ and a characteristic complex amplitude 
    \begin{equation}
        R_i^{(k)}(\omega):=\sum_{\ell=0}^{N-1}\dfrac{ v_k^{[\ell]}v_i^{[\ell]}}{-\omega^2+\imath\alpha\omega+\lambda^{[\ell]}}.
    \label{eq:oscillator_response_factor}
    \end{equation}
    for node $i$.
\end{proposition}

\begin{proof}
The steady-state response to a sinusoidal perturbation at a single node $k$ is obtained by studying the asymptotic behaviour of the responses (\ref{eq:LRT_oscillator_sinusoidal}) has the form of
\begin{equation}
    \bm{\Theta}^{(k)}(t)\overset{t\rightarrow\infty}{\sim}\dfrac{\imath\varepsilon e^{\imath\varphi}}{\alpha\omega N}\bm{1}+e^{\imath(\omega t+\varphi)}\sum_{\ell=0}^{N-1}\dfrac{\varepsilon v_k^{[\ell]}}{-\omega^2+\imath\alpha\omega+\lambda^{[\ell]}}\bm{v}^{[\ell]},
    \label{eq:oscillator_steady_sinusoidal}
\end{equation}
which is composed of a homogeneous phase shift $\frac{\imath\varepsilon e^{\imath\varphi}}{\alpha\omega N}$ and a driven oscillation at each node. Each node's angular variable $\theta_i$ changes at the same frequency as the perturbation frequency $\omega$, but with a complex amplitude
\begin{equation}
    R_i^{(k)}(\omega)=\sum_{\ell=0}^{N-1}\dfrac{ v_k^{[\ell]}v_i^{[\ell]}}{-\omega^2+\imath\alpha\omega+\lambda^{[\ell]}}.
\end{equation}
\end{proof}

\begin{remark}[Characterization of the steady-state response patterns to a sinusoidal perturbation]
    The complex nature of the amplitude suggests shifts in the amplitude and in the phase between the perturbation signal and the response signal, which are both topology-dependent and node-specific. Hence we also refer to $R_i^{(k)}$ as the nodal response factor of node $i$ to a sinusoidal perturbation at node $k$. 

    The homogeneous phase shift contributes to neither the change of grid frequency nor the overall power flow pattern in the network, while the absolute value of $R_i^{(k)}$ determines the maximal deviation of the local grid frequency at node $i$ caused by a perturbation at node $k$ through $|\delta\dot{\theta}_i|=|\dot{\Theta}^{(k)}_i|=\varepsilon\omega|R_i^{(k)}|$. If it exceeds the safety range of normal operation, the local frequency deviation may cause damage the synchronous machine and other related grid components such as the turbine. In the rest of the subsection we focus on the steady-state response pattern constituted by the set of nodal response strengths 
    \begin{equation}
        A_i^{(k)}:=\omega\left|R_i^{(k)}\right|
        \label{eq:nodal_response_strength}
    \end{equation}
    for each node $i$ and discuss in detail its distinctive spatial distributions in different frequency regimes.
\end{remark}

\begin{figure}[ht!]
    \centering
    \includegraphics[width=0.8\columnwidth]{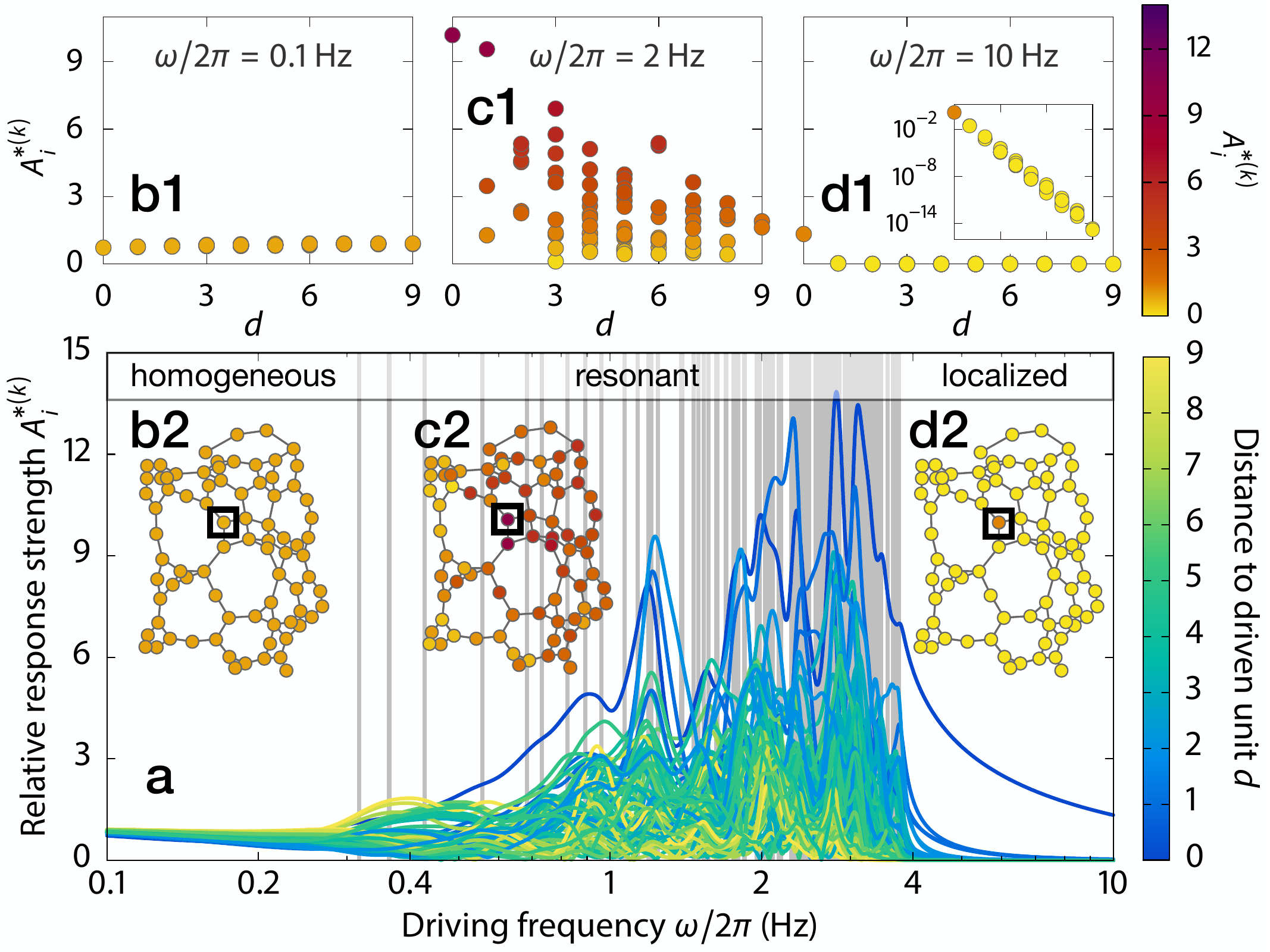}
    \caption{\textbf{Steady-state response patterns exhibit three frequency regimes.} (a) Relative response strength $A_i^{*(k)}$ (see \autoref{remark:relative_response_strength} for definition) of all $80$ nodes in an example power grid network across all three frequency regimes (homogeneous bulk, resonant, and localized responses). Vertical grey lines represent the $N-1$ eigenfrequencies. (b1,c1,d1) Qualitatively different dependencies of $A_i^{*(k)}$ on the graph-theoretic distance $d:=d(k,i)$ with three representative driving frequencies $\omega/2\pi\in\{0.1,2,10\}$ Hz of three frequency regimes. The exponential dependence of $A_i^{*(k)}$ on $d$ is illuatrated in the inset of d1. (b2,c2,d2) Distinctive response patterns for the three driving frequencies, corresponding to (b1,c1,d1). The curves in (a) are color-coded with the distance $d$, and the discs in (b1-b2,c1-c2,d1-d2) are color-coded with the relative response strength $A_i^{*(k)}$. The black square marks the perturbed node. Network settings are the same as Fig.~2 in \citeasnoun{zhang_fluctuation-induced_2019}.}
    \label{fig:regimes}
\end{figure}

\subsubsection{Regime of grid resonances}

As suggested in \autoref{remark:oscillator_eigenfrequencies}, the dynamics of the perturbed oscillator model of AC power grids can be understood in comparison with the dynamics of a driven damped harmonic oscillator. For each intrinsic under-damped oscillation mode corresponding to a non-zero Laplacian eigenvalue $\lambda^{[\ell]}>\alpha^2/4$, the oscillatory power grid system resonates if the perturbation frequency matches the corresponding eigenfrequency $\omega_{\text{eigen}}^{[\ell]}=\sqrt{\lambda^{[\ell]}-\alpha^2/4}$. Driven at frequencies close to an eigenfrequency $\omega_{\text{eigen}}^{[\ell]}$, the network responses exhibit large amplitudes since the rationalized denominator $(\omega^2-\lambda^{[\ell]})^2-\alpha^2\omega^2$ for the corresponding $\ell$-th oscillation mode is minimized. However, the response amplitudes vary greatly for different nodes in the network due to the factor $v_k^{[\ell]}v_i^{[\ell]}$. If the system dissipation is sufficiently low such that there exist $N-1$ under-damped oscillation modes, the corresponding $N-1$ eigenfrequencies form a resonance regime where the power grid system can potentially exhibit strong distributed responses across the network (cf.~\autoref{remark:oscillator_eigenfrequencies}).

We emphasize that the spatiotemporal resonance pattern is characteristic of each perturbation frequency within the resonance regime, of each specific network topology including the prior perturbation base state, and of each location of perturbation. Therefore a power grid system's responses to real-world fluctuating perturbations containing a collection of frequency components within the resonance regime are quite irregular both temporally and spatially, which makes it a non-trivial task to evaluate the risks in perturbed power grids induced by network resonances (see Sec.~\ref{subsec:dvi} for further discussions). 

\subsubsection{Homogeneous responses: the low-frequency regime}

The network response pattern for lower perturbation frequencies, i.e. the ones lower than the smallest eigenfrequency $\omega_{\text{eigen}}^{[1]}=\sqrt{\lambda^{[1]}-\alpha^2/4}$, can be understood by investigating the asymptotic behaviour of the nodal response strength $A_i^{(k)}=\omega|R_i^{(k)}|$ as the $\omega\rightarrow 0$.

\begin{proposition}[Homogeneous responses at the low frequency limit]
    As the perturbation frequency $\omega\rightarrow0$, the steady-state response strength at each node of an AC power grid system with an arbitrary topology (\ref{eq:oscillator_model}) approaches a constant value, i.e.
    \begin{equation}
        A_i^{(k)}\overset{\omega\rightarrow 0}{\sim}\frac{1}{N\alpha}.
        \label{eq:low_frequency_limit}
    \end{equation}
\end{proposition}

\begin{proof}
    Considering \eqref{eq:oscillator_response_factor} and $\lambda^{[0]}=0$, we derive the asymptotic behaviour of the real part and the imaginary part of the nodal response factor $R_i^{(k)}$ as $\omega\rightarrow 0$:
    \begin{align}
        \Real\left[R_i^{(k)}\right]&=\sum_{\ell=0}^{N-1}\frac{v^{[\ell]}_kv^{[\ell]}_i\left(-\omega^2+\lambda^{[\ell]}\right)}{\left(-\omega^2+\lambda^{[\ell]}\right)^2+\alpha^2\omega^2}\overset{\omega\rightarrow 0}{\sim}-\dfrac{1}{N\alpha^2}+\displaystyle\sum_{\ell=1}^{N-1}\dfrac{v^{[\ell]}_kv^{[\ell]}_i}{\lambda^{[\ell]}}\label{eq:real_response_factor_low_freq}\\
        \Imag\left[R_i^{(k)}\right]&=\sum_{\ell=0}^{N-1}\frac{v^{[\ell]}_kv^{[\ell]}_i\left(-\alpha\omega\right)}{\left(-\omega^2+\lambda^{[\ell]}\right)^2+\alpha^2\omega^2}\overset{\omega\rightarrow 0}{\sim}-\dfrac{1}{N\alpha\omega}.
        \label{eq:imag_response_factor_low_freq}
    \end{align}
As the asymptotic behaviour of the response strength $A_i^{(k)}=\omega|R_i^{(k)}|$ is dominated by the imaginary part, we have
    \begin{equation}
        A_i^{(k)}=\omega\left|R_i^{(k)}\right|\overset{\omega\rightarrow 0}{\sim}\omega\cdot\dfrac{1}{N\alpha\omega}=\dfrac{1}{N\alpha}.
    \end{equation}
\end{proof}

\begin{remark}[Consistency with the homogeneous grid frequency shift at $\omega=0$]
    The homogeneous response strength at each node as $\omega\rightarrow 0$ suggests a global shift of grid frequency inversely proportional to the network size $N$ and the system dissipation parameter $\alpha$. This result is quantitatively consistent with the homogeneous grid frequency shift induced by constant perturbations, as discussed in \autoref{prop:osillator_pattern_constant} with $D_i^*=\delta_{ik}$ if node $k$ is the perturbed one.
\end{remark}

\begin{remark}[Relative nodal response strength]
\label{remark:relative_response_strength}
    The homogeneous nodal response at the low frequency limit (\ref{eq:low_frequency_limit}) serves as a reference value for the nodal response strengths of a sinusoidally driven network. Therefore, by normalizing $A_i^{(k)}$ with its low frequency limit $\frac{1}{N\alpha}$, we define the relative nodal response strength
    \begin{align}
        A_i^{*(k)}:=\dfrac{A_i^{(k)}}{\displaystyle\lim_{\omega\rightarrow0}A_i^{(k)}}=N\alpha A_i^{(k)}
        \label{eq:relative_response_strength}
    \end{align}
    so that the nodal response strengths can be compared across networks with different sizes and dissipation values.
\end{remark}

\subsubsection{Localized responses: the high-frequency regime}
\label{subsubsec:localized_regime}

Similarly, we investigate the network response pattern in the high-frequency regime where $\omega>\omega_{\text{eigen}^{[N-1]}}$ by observing the asymptotic behaviour of the nodal response strengths $A_i^{(k)}$ as the perturbation frequency becomes sufficiently large and approaches infinity.

\begin{proposition}[Localized response patterns in the high-frequency regime]
    \label{prop:localized_pattern}
    As the perturbation frequency $\omega\rightarrow\infty$, the steady-state nodal response strength $A_i^{(k)}$ in an AC power grid system with an arbitrary topology (\ref{eq:oscillator_model}) decays as a power-law of $\omega$ with an exponent depending on the graph-theoretic distance $d$ between node $k$ and $i$, i.e.
    \begin{equation}
        A_i^{(k)}\overset{\omega\rightarrow\infty}{\sim}\left|\Phi^{[d]}_{ki}\right|\omega^{-2d-1},
		\label{eq:localized_scaling}
    \end{equation}
    where $\left|\Phi^{[d]}_{ki}\right|$ is a distance- and node-specific prefactor but independent on the perturbation frequency.
\end{proposition}

\begin{proof}
    To determine the asymptotic behaviour of the response strength $A_i^{(k)}=\omega|R_i^{(k)}|$, we first reduce $\Real[R_i^{(k)}]$ and $\Imag[R_i^{(k)}]$ from \eqref{eq:real_response_factor_low_freq} and \eqref{eq:imag_response_factor_low_freq} to a common denominator $M(\omega)$ and obtain the respective numerators $N_{\text{Re}}(\omega)$ and $N_{\text{Im}}(\omega)$ as polynomials of $\omega$,
    \begin{align}
        M(\omega)&:=\displaystyle\prod_{\ell=0}^{N-1}\left[\left(-\omega^2+\lambda^{[\ell]}\right)^2+\alpha^2\omega^2\right],\\
        N^{\text{Re}}_{ki}(\omega)&:=\displaystyle\sum_{\ell=0}^{N-1}v_k^{[\ell]}v_i^{[\ell]}\left(-\omega^2+\lambda^{[\ell]}\right)Q^{[\ell]}(\omega), \text{ and}\label{eq:numerator_Re}\\
        N^{\text{Im}}_{ki}(\omega)&:=\displaystyle\sum_{\ell=0}^{N-1}v_k^{[\ell]}v_i^{[\ell]}\left(-\alpha\omega\right)Q^{[\ell]}(\omega) \quad\text{with}\label{eq:numerator_Im}\\
        Q^{[\ell]}(\omega)&:=\prod_{\substack{\ell'=0,\ell'\neq \ell}}^{N-1}\left[\left(-\omega^2+\lambda^{[\ell']}\right)^2+\alpha^2\omega^2\right].
    \end{align}
    The asymptotic behaviour of $M(\omega)$, $N^{\text{Re}}_{ki}(\omega)$ and $N^{\text{Im}}_{ki}(\omega)$ as $\omega\rightarrow\infty$ is dominated by the respective leading terms with the highest power of $\omega$. The denominator scales asymptotically as
    \begin{equation}
        M(\omega)\overset{\omega\rightarrow\infty}{\sim}\omega^{4N}.
        \label{eq:denominator_asymptotic}
    \end{equation}
    For the numerators the leading term depends on a common product $Q^{[\ell]}(\omega)$. As shown in \autoref{sec:proof_Q}, \st{$Q^{[\ell]}(\omega)$ explicitly depends on $\lambda^{[\ell]}$ in the form of} $Q^{[\ell]}(\omega)$ can be written in terms of $\lambda^{[\ell]}$ and other variables that are dependent on the underlying matrix $\mathcal{L}$ but independent of the choice of $\ell$. Thus, we define
    \begin{equation}
        Q(\lambda^{[\ell]},\omega):=\sum_{j=0}^{2N-2}C^{[j]}(\lambda^{[\ell]})\omega^{4N-4-2j}
        \label{eq:polynomial_Q}
    \end{equation}
     where the coefficients $C^{[j]}(\lambda^{[\ell]})$ are polynomials in $\lambda^{[\ell]}$ of degree $j$. Inserting the expression of $Q(\lambda^{[\ell]},\omega)$ (\ref{eq:polynomial_Q}) to the numerators (\ref{eq:numerator_Re}) and (\ref{eq:numerator_Im}), we obtain
    \begin{align}
        N^{\text{Re}}_{ki}&=\sum_{\ell=0}^{N-1}v_k^{[\ell]}v_i^{[\ell]}\sum_{j=0}^{2N-1}F^{[j]}(\lambda^{[\ell]})\omega^{4N-2-2j}\text{ and }
	    N^{\text{Im}}_{ki}&=\sum_{\ell=0}^{N-1}v_k^{[\ell]}v_i^{[\ell]}\sum_{j=0}^{2N-2}G^{[j]}(\lambda^{[\ell]})\omega^{4N-3-2j}
	    \label{eq:numerator_polynomials}
    \end{align}
    where $F^{[j]}(\lambda^{[\ell]})$ and $G^{[j]}(\lambda^{[\ell]})$ are also polynomials in $\lambda^{[\ell]}$ of degree $j$ and can be written in terms of $C^{[j]}(\lambda^{[\ell]})$ as
    \begin{align}
        F^{[j]}(\lambda^{[\ell]})&=
        \begin{cases}
            -C^{[j]}(\lambda^{[\ell]})   & j=0\\
            -C^{[j]}(\lambda^{[\ell]})+\lambda^{[\ell]}C^{j-1}(\lambda^{[\ell]})  & 1\leq j \leq 2N-2\\
            \lambda^{[\ell]}C^{j-1}(\lambda^{[\ell]}) & j=2N-1
        \end{cases}  \label{eq:coeff_real_numerator}\\
        G^{[j]}(\lambda^{[\ell]})&=-\alpha C^{[j]}(\lambda^{[\ell]}),\quad j\in\{0,1,\cdots,2N-2\}.
    \end{align}
    Considering the numerators $N^{\text{Re}}_{ki}(\omega)$ and $N^{\text{Im}}_{ki}(\omega)$ as the $ki$-th elements of numerator matrices $\mathcal{N}^{\text{Re}}$ and $\mathcal{N}^{\text{Im}}$, we can conveniently write (\ref{eq:numerator_polynomials}) in a matrix form
    \begin{equation}
        \mathcal{N}^{\text{Re}}=\sum_{j=0}^{2N-1}{\Phi}^{[j]}\omega^{4N-2-2j},\quad
	    \mathcal{N}^{\text{Im}}=\sum_{j=0}^{2N-2}{\Gamma}^{[j]}\omega^{4N-3-2j}
    \end{equation}
    with the coefficient matrices ${\Phi}^{[j]}:=\mathcal{V}\mathcal{F}^{[j]}\mathcal{V}^{\text{T}}$ and ${\Gamma}^{[j]}:=\mathcal{V}\mathcal{G}^{[j]}\mathcal{V}^{\text{T}}$. Here $\mathcal{V}:=(v^{[0]},\cdots,v^{[N-1]})$ and $\mathcal{F}^{[j]}$, $\mathcal{G}^{[j]}$ are diagonal matrices with $\mathcal{F}^{[j]}_{ii}:=F^{[j]}(\lambda^{[i]})$ and $\mathcal{G}^{[j]}_{ii}:=G^{[j]}(\lambda^{[i]})$, respectively. By spelling out the polynomials \begin{equation}
        F^{[j]}(\lambda^{[\ell]})=\sum_{m=0}^jf^{[j]}_m\cdot(\lambda^{[\ell]})^m,\quad
        G^{[j]}(\lambda^{[i]})=\sum_{m=0}^jg^{[j]}_m\cdot(\lambda^{[\ell]})^m
    \end{equation}
    with coefficients $f_m^{[j]},g_m^{[j]}\in\mathbb{R}$, we can see the diagonal matrices $\mathcal{F}^{[j]}$ and $\mathcal{G}^{[j]}$ are polynomials of a diagonal matrix $\Lambda$ with $\Lambda_{ii}:=\lambda^{[i]}$
    \begin{equation}
        \mathcal{F}^{[j]}=\sum_{m=0}^jf^{[j]}_m\Lambda^m,\quad
        \mathcal{G}^{[j]}=\sum_{m=0}^jg^{[j]}_m\Lambda^m,
    \end{equation}
    so that the coefficient matrices $\Phi^{[j]}$ and $\Gamma^{[j]}$ can be written as
    \begin{align}
        {\Phi}^{[j]}&=\mathcal{V}\left(\sum_{m=0}^{j}f^{[j]}_m\Lambda^m\right)\mathcal{V}^{\text{T}}=\sum_{m=0}^{j}f^{[j]}_m\left(\mathcal{V}\Lambda^m\mathcal{V}^{\text{T}}\right)=\sum_{m=0}^{j}f^{[j]}_m\mathcal{L}^m,\label{eq:coeff_matrix_phi}\\
        {\Gamma}^{[j]}&=\mathcal{V}\left(\sum_{m=0}^{j}g^{[j]}_m\Lambda^m\right)\mathcal{V}^{\text{T}}=\sum_{m=0}^{j}g^{[j]}_m\left(\mathcal{V}\Lambda^m\mathcal{V}^{\text{T}}\right)=\sum_{m=0}^{j}g^{[j]}_m\mathcal{L}^m,\label{eq:coeff_matrix_gamma}
    \end{align}
    indicating that they are polynomials of the weighted graph Laplacian matrix $\mathcal{L}$ of degree $j$, as $\mathcal{V}\Lambda\mathcal{V}^{\text{T}}\equiv\mathcal{L}$. In short, the numerators $\mathcal{N}^{\text{Re}}$ and $\mathcal{N}^{\text{Im}}_{ki}$ are in fact polynomials in $\omega$, with its coefficient being also a polynomial in $\mathcal{L}$,
    \begin{equation}
        \mathcal{N}^{\text{Re}}=\sum_{j=0}^{2N-1}\left(\sum_{m=0}^{j}f^{[j]}_m\mathcal{L}^m\right)\omega^{4N-2-2j},\quad
	    \mathcal{N}^{\text{Im}}=\sum_{j=0}^{2N-2}\left(\sum_{m=0}^{j}g^{[j]}_m\mathcal{L}^m\right)\omega^{4N-3-2j}.
    \end{equation}
    We note that as the index $j$ increases, the powers of $\omega$ decrease and the degrees of the polynomials $\Phi^{[j]}(\mathcal{L})$ and $\Gamma^{[j]}(\mathcal{L})$, i.e. the coefficients of $\omega$, increase. For sufficiently large perturbation frequency $\omega$, the leading terms in the numerators which dominate the asymptotic behaviors would be the ones with the highest degrees of $\omega$ with nonzero coefficients. We know from graph theory that $(\mathcal{L}^m)_{ij}\neq0$ only for node pair $(i,j)$ with the graph theoretic distance  $d(i,j)$ between them, i.e. the length of the shortest path from $j$ to $i$ on the unweighted graph defined by $\mathcal{L}$,  satisfying $d(i,j)\leq m$. Therefore, the first terms in $N^{\text{Re}}_{ki}(\omega)$ and $N^{\text{Im}}_{ki}(\omega)$ with the highest degrees of $\omega$ have exactly zero coefficients, i.e. $\Phi^{[j]}=0$ and $\Gamma^{[j]}=0$, for all $j<d(k,i)$. Consequently, the leading term in the numerators are
    \begin{equation}
        N^{\text{Re}}_{ki}(\omega)\overset{\omega\rightarrow\infty}{\sim}\Phi^{[d]}_{ki}\omega^{4N-2-2d},\quad
        N^{\text{Im}}_{ki}(\omega)\overset{\omega\rightarrow\infty}{\sim}\Gamma^{[d]}_{ki}\omega^{4N-3-2d}
    \end{equation}
    with $d:=d(k,i)$. Together with the asymptotic behavior of the denominator (\ref{eq:denominator_asymptotic}), we obtain 
    \begin{equation}
        A_i^{(k)}=\omega\left|R_i^{(k)}\right|\overset{\omega\rightarrow \infty}{\sim}\omega\cdot\left|\dfrac{\Phi^{[d]}_{ki}\omega^{4N-2-2d}}{\omega^{4N}}\right|=\left|\Phi^{[d]}_{ki}\right|\omega^{-2d-1}.
    \end{equation}
\end{proof}

\begin{remark}[Localized response patterns in the high-frequency limit]
    \autoref{prop:localized_pattern} implies that the response amplitude of the grid frequency $A_i^{(k)}$ decays exponentially with the graph-theoretic distance $d$ between the perturbed node and the responding node. The response amplitude also decays as a power law in $\omega$ for fixed $d$. In the high-frequency limit, a network's response to the perturbation is restricted to the perturbed node, i.e.
    \begin{equation}
        \lim_{\omega\rightarrow\infty}\dfrac{A_i^{(k)}}{A_i^{(k)}}
        \overset{\omega\rightarrow \infty}{\sim}\lim_{\omega\rightarrow\infty}\dfrac{\left|\Phi^{[d]}_{ki}\right|\omega^{-2d-1}}{\left|\Phi^{[0]}_{ki}\right|\omega^{-1}}
        =\lim_{\omega\rightarrow\infty}\left|\Phi^{[d]}_{ki}\right|\omega^{-2d}=\delta_{ki}
    \end{equation}
    with $\delta_{ki}$ being the Kronecker delta function.
\end{remark}

\begin{figure}[!ht]
    \centering
    \includegraphics[width=0.8\columnwidth]{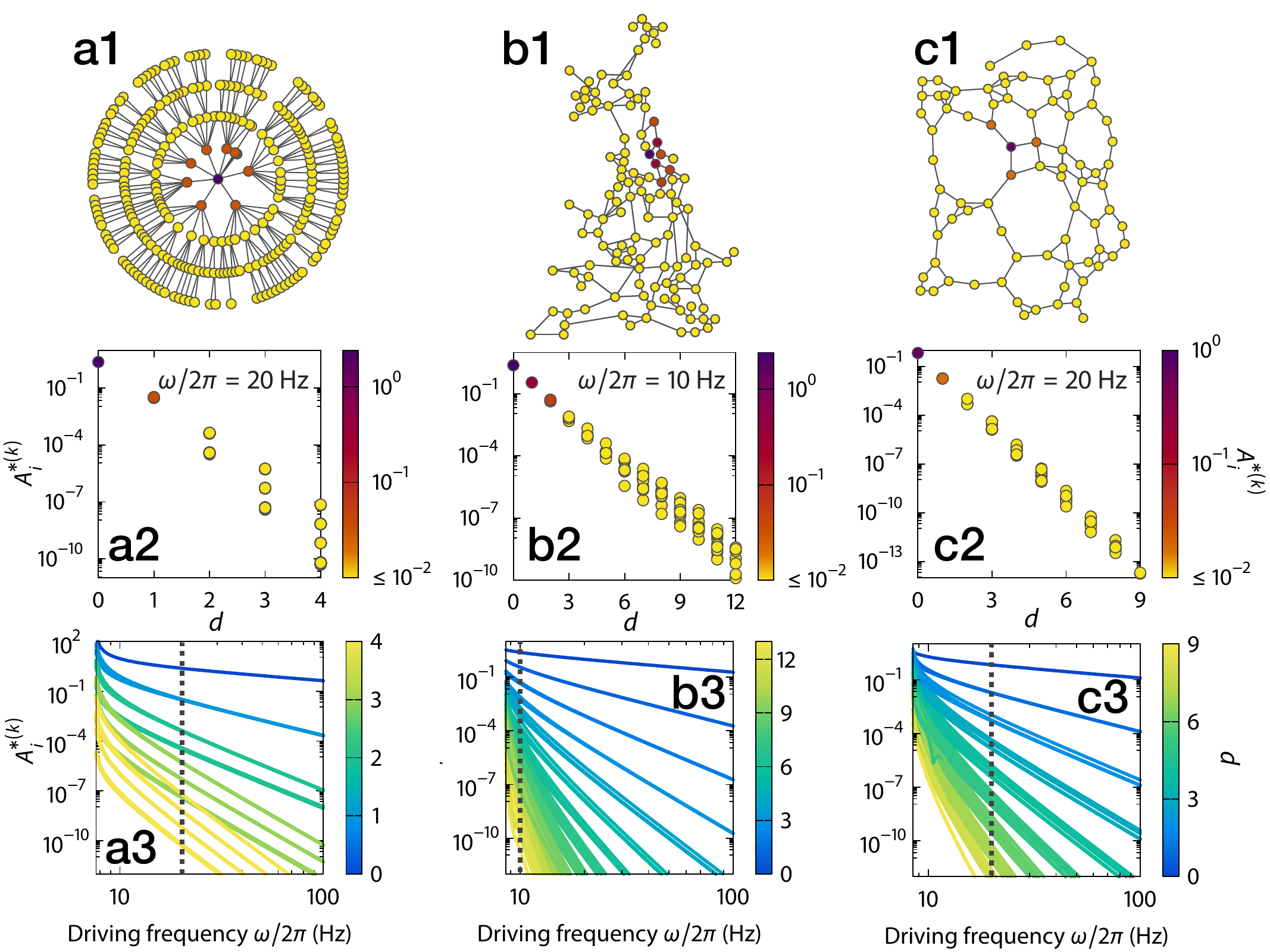}
    \caption{\textbf{Topological localization of network responses.} For frequencies larger than all eigenfrequencies and across network types (row 1) the response amplitudes (\ref{eq:relative_response_strength}) decays exponentially with shortest-path distance $d$ (row 2) and algebraically with driving frequency $\omega$ (row 3) (cf.~\autoref{prop:localized_pattern}). Dashed vertical lines in row 3 indicate the displayed frequency responses in row 2. Columns display graphs and responses for a random tree (column a), the topology of the British high voltage transmission grid (column b) and a random power grid network topology generated according to \citeasnoun{schultz_random_2014} (column c). Network settings: $(N,N_g,P_g,P_c,K_g,K_c,\alpha)=(264,24,10\t{ s}^{-2},-1\t{ s}^{-2},200\t{ s}^{-2},20\t{ s}^{-2},1\t{ s}^{-1})$ for column a, $(120,30,39\t{ s}^{-2},-13\t{ s}^{-2},390\t{ s}^{-2},390\t{ s}^{-2},1\t{ s}^{-1})$ for column b, and $(80,20,39\t{ s}^{-2},-13\t{ s}^{-2},390\t{ s}^{-2},390\t{ s}^{-2},1\t{ s}^{-1})$ for column c \protect\footnotemark.} 
    \label{fig:localized}
\end{figure}
    
\begin{remark}[Localized response patterns in networks of multi-dimensional dynamical systems]
    We consider networks of $N$ diffusively coupled identical units governed by $n$-dimensional dynamics. Each unit $i$ has $n$ state variables $(x_i^{[0]},x_i^{[1]}\cdots,x_i^{[n-1]})$ and is governed by dynamics
    \begin{align}
        \begin{cases}
            \dot{x}_i^{[0]}&=x_i^{[1]}\nonumber\\
            \dot{x}_i^{[1]}&=x_i^{[2]}\nonumber\\
            \cdots\nonumber\\
            \dot{x}_i^{[n-1]}&=f(x_i^{[0]},x_i^{[1]}\cdots,x_i^{[n-1]})+\displaystyle\sum_{j=1}^N g(x_j^{[0]}-x_i^{[0]}),\nonumber
        \end{cases}
    \end{align}
    where $f:\mathbb{R}^N\rightarrow\mathbb{R}$ and $g:\mathbb{R}\rightarrow\mathbb{R}$ are functions respectively representing the intrinsic and the coupling dynamics of units and allowing for a stable fixed point of the system. If unit $k$ is sinusoidally driven with frequency $\omega\rightarrow\infty$, we conjecture that the amplitude $\tilde{A}_{i,m}^{(k)}$ of the sinusoidal response in state variable $D_t^m x_i$ of unit $i$ is given by
    \begin{equation}
        \tilde{A}_{i,m}^{(k)}\overset{\omega\rightarrow \infty}{\sim}\left|\Psi^{[d]}_{ki}\right|\omega^{-n(d+1)+m},
        \label{conj:localize_pattern}
    \end{equation}
    where $|\Psi^{[d]}_{ki}|$ is a distance- and node-specific prefactor but independent on the perturbation frequency.

    The response of a sinusoidally driven damped harmonic oscillator with $\tilde{A}_{i,m}^{(k)}\overset{\omega\rightarrow \infty}{\sim}\omega^{-2}$ can be seen as a special case of (\ref{conj:localize_pattern}) with $n=2,d=0$ and $m=0$. For networks of Kuramoto phase oscillators with $n=1$, (\ref{conj:localize_pattern}) is proven to be valid [see the Supplementary of \citeasnoun{zhang_fluctuation-induced_2019}]. For the oscillator model of power grid with $n=2$ and $m=1$, (\ref{conj:localize_pattern}) reduces to  \autoref{prop:localized_pattern}.
\end{remark}

\footnotetext{Here in network settings $N_g$ is the number of power generating units with power injection $P_g>0$ and the rest units are power consuming units with $P_c<0$. The transmission lines connected to generators have capacity $K_g$ and the rest lines have capacity $K_c$.}

\begin{remark}[Generalizability of the steady-state response patterns in three frequency regimes]
	In the above discussions of the steady responses patterns in three frequency regimes in Sec.~\ref{subsec:steady_patterns}, we do not make any assumptions on the network topology. Therefore, our results on the characteristics of the homogeneous, the resonance and the localized response patterns in three regimes hold for arbitrary network topologies. Nevertheless, the evaluation of the parameters and the prefactors, such as $\Phi_{ki}^{[d]}$ in \eqref{eq:localized_scaling}, is network- and topology-dependent by definition.
\end{remark}

\subsection{Topological factor of transient spreading dynamics}
\label{subsec:transient_patterns}

In the following we focus on the transient response of AC power grids to external perturbations, which is referred to as the network responses close to the time of perturbation and thus describes the spatiotemporal pattern in the perturbation spreading process across the network. Particularly, we demonstrate how to extract the role of the  network topology in the spreading pattern based on the LRT of the oscillator model.

To investigate the transient response close to the onset of perturbation at $t=0$, we Taylor-expand the linear response (\ref{eq:LRT_oscillator_sinusoidal}) at node $i$ to a sinusoidal perturbation at node $k$ in powers of $t$ as
\begin{equation}
    \Theta_i^{(k)}(t)=\sum_{n=0}^{\infty}\dfrac{D_t^n{\Theta}_i^{(k)}(0)}{n!}t^n
    \label{eq:transient_Taylor_expansion}
\end{equation}
around $t=0$, which is characterized by the time derivatives of the linear response at $t=0$. Here $D^n_t:=\frac{\mathrm{d}^n}{\mathrm{d}t^n}$ is Euler's notation for differential operator. The $n$-th order time derivative of the linear response at $t=0$ is
\begin{equation}
    D^n_t\Theta_i^{(k)}(0)=\sum_{\ell=0}^{N-1}
	\dfrac{v_k^{[\ell]}v_i^{[\ell]}\varepsilon e^{\imath\varphi}}{-\omega^2+\imath\alpha\omega+\lambda^{[\ell]}}
	\left[\dfrac{\left(\Delta^{[\ell]}_+\right)^n\left(\Delta^{[\ell]}_--\imath\omega\right)-\left(\Delta^{[\ell]}_-\right)^n\left(\Delta^{[\ell]}_+-\imath\omega\right)}{2\eta^{[\ell]}}+(\imath\omega)^n\right].
	\label{eq:transient_derivative}
\end{equation}
The transient response of a power grid network can thus be estimated by the first non-zero term in the power series of $t$ (\ref{eq:transient_Taylor_expansion}). Interestingly, the resulting series does not start with low powers of $t$ such as $t^0$ or $t^1$ as typical for common Taylor expansions. Instead, it typically starts with large powers of $t$ as the following proposition illustrates.

\begin{proposition}[Leading-term approximation of transient response]
\label{prop:transient_leading_term}
    The transient response at node $i$ in an AC power grid network (\ref{eq:oscillator_model}) to a sinusoidal perturbation of frequency $\omega$ at node $k$ with an onset at $t=0$ is approximated by the $(2d+2)$-nd term in the Taylor expansion of the linear response (\ref{eq:LRT_oscillator_sinusoidal}) around $t=0$,
    \begin{equation}
        \Theta_i^{(k)}(t)=\dfrac{\varepsilon e^{\imath\varphi}(-1)^{d}\left(\mathcal{L}^{d}\right)_{ki}}{\left(2d+2\right)!}t^{2d+2}+O(t^{2d+3}).
        \label{eq:leading_term_approximation}
    \end{equation}
    Here $d:=d(k,i)$ denotes the graph-theoretic distance between node $k$ and node $i$.
\end{proposition}

\begin{proof}
    To find out the leading term in the Taylor series of the linear response (\ref{eq:transient_Taylor_expansion}), we study the summand of the $\ell$-th eigenmode in the derivative of the linear response $D^n_t\Theta_i^{(k)}(0)$  (\ref{eq:transient_derivative}), which we denote as $v_k^{[\ell]}v_i^{[\ell]}\varepsilon e^{\imath\varphi}F_n(\lambda^{[\ell]})$ for convenience\footnote{Please not that the functions $F_n$ here and $F^{[j]}$ in \eqref{eq:coeff_real_numerator} are different from each other though both are polynomials in $\lambda^{[\ell]}$.}. In the summand, the function $F_n(\lambda^{[\ell]})$ appears to be a division of two polynomials of $\lambda^{[\ell]}$. In fact, it can be shown that the leading term of  $F_n(\lambda^{[\ell]})$ (denoted as $\lt{F_n(\lambda^{[\ell]})}$), i.e. the term with the highest order of $\lambda^{[\ell]}$ is
    \begin{align}
        \lt{F_n(\lambda^{[\ell]})}=\left\lbrace
		  \begin{array}{ll}
			(-1)^{\frac{n-1}{2}}\left(-\imath\omega+\dfrac{n-1}{2}\alpha\right)\left(\lambda^{[\ell]}\right)^{\frac{n-3}{2}}& \quad \text{if } n \text{ is odd,}\\
		   (-1)^{\frac{n-2}{2}}\left(\lambda^{[\ell]}\right)^{\frac{n-2}{2}}& \quad \text{if } n \text{ is even.}
		  \end{array}\right.
		 \label{eq:derivative_leading_term}
    \end{align}
    A derivation of the result (\ref{eq:derivative_leading_term}) is given in \autoref{sec:proof_F}. We note that, for $n=0$ and $n=1$, $F_n(\lambda^{[\ell]})=0$, which is a consequence of the choice of initial condition: the linear response and its first derivative are supposed to be zero at $t=0$, as the responses $\Theta_i^{(k)}$ and the frequency response $\dot{\Theta}_i^{(k)}$ are zero at the onset of perturbation. For $n\geq2$, (\ref{eq:derivative_leading_term}) indicates a monotonic relation between the degree of $F_n(\lambda^{[\ell]})$ as a polynomial of $\lambda^{[\ell]}$ and the order of derivative $n$: $n=2\deg[F_n(\lambda^{[\ell]})]+2$.
    
    As we have shown in Sec.~\ref{subsubsec:localized_regime}, sums of the form of $\sum_{\ell=0}^{N-1}v_k^{[\ell]}v_i^{[\ell]}P^j(\lambda^{[\ell]})$, where $P^j(\lambda^{[\ell]})$ represents a general polynomial of $\lambda^{[\ell]}$ of degree $j$, can be seen as $[P^j(\mathcal{L})]_{ki}$, the $ki$-th element of the matrix $P^j(\mathcal{L})$, a polynomial of $\mathcal{L}$ with degree $j$. Applying this result to the derivatives of the linear response (\ref{eq:transient_derivative}), we find that $D^n_t\Theta_i^{(k)}(0)$ can be considered as the $ki$-th element of matrix $F_n(\mathcal{L})$, a polynomial of $\mathcal{L}$. The leading term of $D^n_t\Theta_i^{(k)}(0)$ is thus given by
    \begin{equation}
        \lt{D^n_t\Theta_i^{(k)}(0)}=\varepsilon e^{\imath\varphi} \lt{F_n(\mathcal{L})}_{ki},
    \end{equation}
    which contains $(\mathcal{L}^m)_{ki}$ with $m=\frac{n-3}{2}$ if $n$ is odd and $m=\frac{n-2}{2}$ if $n$ is even. We notice that for a given node pair $(k,i)$ at distance $d$, we have $(\mathcal{L}^m)_{ki}=0$ for all $m<d$ because no path of length $m<d$ can connect nodes $k$ and $i$. Therefore all terms in the Tayler series ($\ref{eq:transient_Taylor_expansion}$) with the leading term's degree lower than $d$ vanish. The first non-zero term in the series thus equals the leading term of the $(2d+2)$-th derivative of the linear response,
    \begin{equation}
        D^{2d+2}_t\Theta_i^{(k)}(0)=\varepsilon e^{\imath\varphi} \lt{F_{2d+2}(\mathcal{L})}_{ki}=\varepsilon e^{\imath\varphi}(-1)^d(\mathcal{L}^d)_{ki},
    \end{equation}
    because all other terms contain $(\mathcal{L}^m)_{ki}$ with $m<d$ and thus vanish. Taken together, the transient linear response near $t=0$ can be approximated as
    \begin{equation}
        \Theta_i^{(k)}(t)=\sum_{n=2d+2}^{\infty}\dfrac{D_t^n\Theta_i^{(k)}(0)}{n!}t^n=\dfrac{\varepsilon e^{\imath\varphi}(-1)^{d}\left(\mathcal{L}^{d}\right)_{ki}}{\left(2d+2\right)!}t^{2d+2}+O(t^{2d+3}).
    \end{equation}
\end{proof}

\begin{remark}[Topological factor in perturbation spreading]
\label{remark:topological_factor}
    The leading-term approximation of the linear response (\autoref{prop:transient_leading_term}) provides a way to disentangle the impact of various factors on the dynamical spreading process in power grid networks. Specifically, the impact of the specific network topology, including the interaction structure between units and the system's base state at $t=0$, is reflected in the factor $(\mathcal{L}^d)_{ki}$ in the leading-term approximation. It satisfies
    \begin{equation}
        \left(\mathcal{L}^{d}\right)_{ki}=\sum_{\mathcal{P}_{k\rightarrow i}^{d}}\ \prod_{(u,v)\in\mathcal{P}_{k\rightarrow i}^{d}}\mathcal{L}_{uv},
    \end{equation}
    suggesting that it can be interpreted as the product of the edge weights along a shortest path $\mathcal{P}_{k\rightarrow i}^{d}$ between node $k$ and $i$, summed over all shortest paths. This insight provides guidelines for manipulating the perturbation spreading dynamics in power grid networks through changing the underlying topology. Numerical evidences show that the topological factor that revealed by the leading-term approximation also enables a master function approach to accurately predict threshold-crossing arrival times in power grid networks [\citet{zhang_topological_2020}].
\end{remark}

\begin{remark}[Scaling behaviours in transient responses]
    The leading-term approximation of the transient response (\ref{eq:leading_term_approximation}) reveals two scaling behaviours as $t\rightarrow0$ (cf.~Fig.~\ref{fig:spreading}). First, the transient response grows algebraically in time with a distance-dependent exponent: $\Theta_i^{(k)}(t)\sim C_d t^{2d+2}$. Here $C_d:={\varepsilon e^{\imath\varphi}(-1)^{d}\left(\mathcal{L}^{d}\right)_{ki}}/{\left(2d+2\right)!}$ is a time independent prefactor but depends on signal magnitude, topology, base operating state and inter-node distance $d$. Second, the transient response decays nearly exponentially with distance $d$, since the factor $t^{2d+2}$ dominates the asymptotic behaviour of the response at large but finite distances as $t\rightarrow0$.
\end{remark}

\begin{figure}[ht]
    \centering
    \includegraphics[width=0.8\columnwidth]{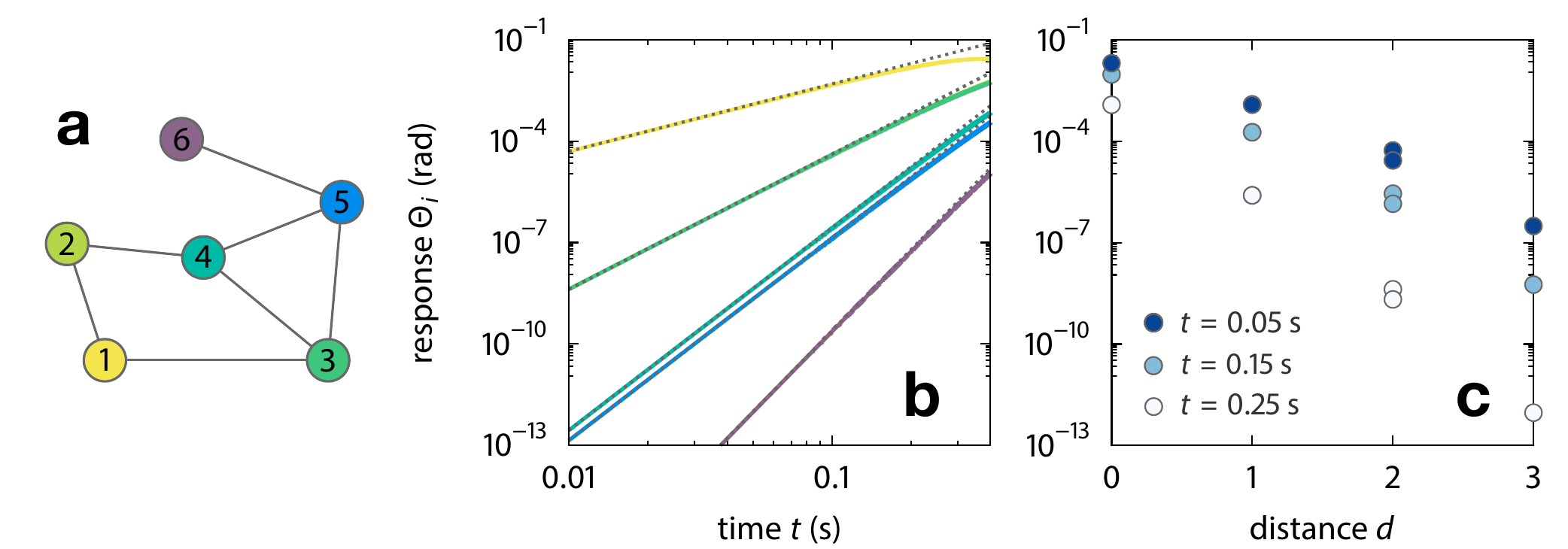}
    \caption{\textbf{Transient Network Response Dynamics} exhibits algebraic growth with time and exponential decay with shortest-path distance. 
    (a) Basic network of $N=6$ units illustrates (b,c) transient algebraic responses [color-coded as units in (a)] to a sinusoidal perturbation at node $1$ that increase like $\Theta_i^{(k)}(t)\sim C_d t^{2d+2}$ as $t \rightarrow 0 $, with time independent constant $C_d$ that depends on signal magnitude, topology, base operating state and inter-node distance $d$, see \eqref{eq:leading_term_approximation}. Thus, responses (b) algebraically increase with time $t$ at any given unit and (c) at any given time, they decay nearly exponentially with shortest-path distance $d=d(k,i)$ between the perturbed unit $k$ and the observed unit $i$. The grey dotted lines in (b) indicate the leading-term approximations. Network settings: $(N,N_g,P_g,P_c,K_g,K_c,\alpha)=(6,3,1\t{ s}^{-2},-1\t{ s}^{-2},10\t{ s}^{-2},10\t{ s}^{-2},1\t{ s}^{-1})$. For the perturbation signal $(\varepsilon,\omega/2\pi,\varphi)=(1,1\t{ Hz},0\t{ rad})$.}
    \label{fig:spreading}
\end{figure}

\begin{remark}[Generalizability of the transient spreading patterns]
	The above discussions on the spatio-temporal pattern of transient spreading do not involve any assumptions on the underlying network topology. Thus the form of the leading term approximation of the transient network response (Proposition \ref{prop:transient_leading_term}) does not depend on the specific choice of network topology. We underline that the evaluation of the topological factor in perturbation spreading (Remark \ref{remark:topological_factor}) is indeed topology- and node-specific, as it captures the local interaction structure and consists of all shortest paths between the perturbed node and the responding node.
\end{remark}

\subsection{Nodal vulnerability to unpredictable fluctuations}
\label{subsec:dvi}

In Sec.~\ref{subsec:steady_patterns} and Sec.~\ref{subsec:transient_patterns} we show how the distributed response patterns of power grid systems, in the steady-state ($t\rightarrow\infty$) and in the transient stage ($t\rightarrow0$), are analytically extracted from the LRT of the oscillator model. The response patterns are numerically proven to be highly accurate [\citet{zhang_fluctuation-induced_2019,zhang_topological_2020}], but the results are valid only for given perturbation signals, hence deterministic. Meanwhile, real-world power grid systems are perturbed by power fluctuations whose exact time series are hardly predictable. In this section we discuss how LRT helps to estimate network responses to random perturbations.

As discussed in Sec.~\ref{subsec:steady_patterns}, power grid systems respond resonantly to perturbations with frequencies falling in the band of network's eigenfrequencies $I_{\t{res}}:=[\omega_{\t{eigen}}^{[1]},\omega_{\t{eigen}}^{[N-1]}]$, exhibiting the most irregular network-wide spatiotemporal patterns, compared to the almost homogeneous pattern for lower frequencies and the localized pattern for higher frequencies. Moreover, the resonance response pattern varies drastically for different perturbation frequencies and for different locations of perturbation. Therefore, estimating the nodal responses for random perturbation signals involving frequency components within $I_{\t{res}}$ is a task not only of practical significance regarding the operational safety of power grid systems, but also with a high theoretical complexity.

\begin{definition}[Dynamic vulnerability index (DVI) for random network resonances]
\label{def:dvi}
    In an AC power grid system (\ref{eq:oscillator_model}) perturbed by a random fluctuation at node $k$, characterized by a power spectral density $S(\omega)$ with frequency components $\omega\in I_{\mathrm{res}}=[\omega_{\mathrm{eigen}}^{[1]},\omega_{\mathrm{eigen}}^{[N-1]}]$, the nodal Dynamic Vulnerability Index (DVI) is defined as
    \begin{equation}
        \mathrm{DVI}_i^{(k)}:=\int_{I_{\mathrm{res}}}{S(\omega)}^{\frac{1}{2}}\left|\sum_{\ell=0}^{N-1}\dfrac{\imath\omega v^{[\ell]}_kv^{[\ell]}_i}{-\omega^2+\imath\alpha\omega+\lambda^{[\ell]}}\right|\ud\omega.
        \label{eq:dvi}
    \end{equation}
\end{definition}

\begin{proposition}[DVI estimates ranking of the nodal all-time-high steady frequency responses]
    Let an AC power grid system (\ref{eq:oscillator_model}) be perturbed by a time-dependent fluctuation at node $k$. A random signal time series is characterized by a power spectral density $S(\omega)$, i.e., the strength $\varepsilon$ of its frequency component $\varepsilon e^{\imath(\omega t+\varphi)}$ is frequency dependent and follows $\varepsilon(\omega)\propto S(\omega)^{\frac{1}{2}}$, and the corresponding phase $\varphi$ is randomly drawn from the uniform distribution on $[0,2\pi)$, independently for each realization of the fluctuation time series. Suppose the fluctuation signal is composed of frequency components with $\omega\in I_{\mathrm{res}}=[\omega_{\mathrm{eigen}}^{[1]},\omega_{\mathrm{eigen}}^{[N-1]}]$, the ranking $\sigma_{\mathrm{ATH}}$ of the nodal all-time-high steady frequency response magnitude $\max_{t\in[0,T]}|\dot{\Theta}^{(k)}_i(t)|$ in an observation window $T$ approaches the ranking $\sigma_{\mathrm{DVI}}$ of the DVI defined in \autoref{def:dvi} as the observation time window goes to infinity, i.e.
    \begin{equation}
        \lim_{T\rightarrow\infty}\sigma_{\mathrm{ATH}}(i)=\sigma_{\mathrm{DVI}}(i)
        \quad \text{for all }i\in\{1,\cdots,N\}.
    \end{equation}
\end{proposition}

\begin{proof}
    To analyse the steady network response, i.e. the response in a time window $T\rightarrow\infty$, to a perturbation signal composed of a range of frequencies, we make use of the steady-state response of the oscillator model to a sinusoidal signal (\ref{eq:oscillator_steady_sinusoidal}). The steady-state response $\dot{\Theta}_i(t)$ of the frequency at node $i$ to a perturbation signal $\varepsilon e^{\imath(\omega t+\varphi)}$ at node $k$ is given by
    \begin{equation}
        \dot{\Theta}_i^{(k)}(t)\overset{t\rightarrow\infty}{\sim}\sum_{\ell=0}^{N-1}\dfrac{\imath\varepsilon\omega v_k^{[\ell]}{v}_i^{[\ell]}}{-\omega^2+\imath\alpha\omega+\lambda^{[\ell]}}e^{\imath(\omega t+\varphi)}=\imath\varepsilon\omega R_i^{(k)}e^{\imath(\omega t+\varphi)},
        \label{eq:steady_frequency_response_sinusoidal}
    \end{equation}
    which can be seen as a driven oscillation with a complex amplitude $\imath\varepsilon\omega R_i^{(k)}$. Here $R_i^{(k)}$ is the response factor defined in (\ref{eq:oscillator_response_factor}) characterizing the response strength at individual nodes in a network. The complex amplitude gives rise to an amplitude shift $|\imath\varepsilon\omega R_i^{(k)}|$ and a phase shift $\arg(\imath\varepsilon\omega R_i^{(k)})$, which both are specific to the perturbation strength $\varepsilon$ and the perturbation frequency $\omega$.
    
    As a consequence of the linear nature of LRT, the nodal frequency response to a temporally fluctuating perturbation signal containing a spectrum of frequency components is obtained by summing the response $\dot{\Theta}_i^{(k)}(\varepsilon,\omega,t)$ (\ref{eq:steady_frequency_response_sinusoidal}) over all frequency components (see also \autoref{remark:distributed_arbitrary_perturbations}). Particularly, for perturbation signals which are characterized by a specific power spectral density (PSD) $S(w)$, the strength of its frequency component $\varepsilon e^{\imath(\omega t+\varphi)}$ can be expressed in terms of the frequency as $\varepsilon(\omega)\propto S(\omega)^{\frac{1}{2}}$ while no assumptions is made on the choice of its phase $\varphi$. For instance, in modern power grids integrated with renewable energies, the power fluctuations from wind and solar energy are characterized by a power law PSD with the Kolmogorov exponent $-5/3$ [\citet{anvari_short_2016}]. For such resonant perturbations, the nodal frequency response reads
    \begin{equation}
        \dot{\Theta}_i^{(k)}(S(\omega),t)\overset{t\rightarrow\infty}{\sim}\displaystyle\int_{I_{\text{res}}}\imath cS(\omega)^{\frac{1}{2}}\omega R_i^{(k)}e^{\imath(\omega t+\varphi)}\ud\omega.
        \label{eq:steady_stochastic_resonance}
    \end{equation}
    Here $c$ is a scaling factor between the power distribution over frequencies of a specific signal and the (normalized) PSD. Independent of the specific realization of the fluctuation signal, the largest possible magnitude of the frequency response (\ref{eq:steady_stochastic_resonance}) is reached only when the involved frequencies are finite in number and non-resonant to each other\footnote{In this specific context we adopt the definition of resonance between frequencies in ergodic theory. The frequencies, as elements of a vector $\bm{\omega}$, are called \emph{resonant} to each other if there exists a nonzero integer vector $\mathbf{m}\cdot\bm{\omega}=0$. Otherwise the frequencies are called \textit{non-resonant} to each other.} so that all oscillating responses for each frequency component $\omega$ (\ref{eq:steady_frequency_response_sinusoidal}) would eventually align. The largest possible magnitude of frequency response time series is given by the sum over $\omega$ of the magnitude of each frequency response $|\imath cS(\omega)^{\frac{1}{2}}\omega R_i^{(k)}|$. For finite time series with $M$ data points and time step $\Delta T$, its frequency components $\omega_n=\frac{2\pi n}{M\Delta T}$ with $n \in \{1,\cdots,\frac{M}{2}\}$ given by discrete Fourier transform are apparently resonant to each other. However, for time series with a fixed sampling rate, the longer the observation time window $T$, the smaller the frequency interval $\Delta \omega=2\pi/T$ and thus the more frequency components exist within the given interval of interest $I_{\text{res}}$. As $T$ approaches infinity, the order of the resonant frequencies $\{\omega_n\}$ becomes sufficiently high so that the alignment of the phases can be attained at a finite rate \cite{dumas_ergodization_1991}. Therefore, as $T\rightarrow\infty$, the all-time-high frequency response approaches the sum of the frequency-specific response magnitudes, which is proportional to the DVI given in (\ref{eq:dvi}):
    \begin{equation}
        \lim_{T\rightarrow\infty}\max_{t\in[0,T]}\left|\dot{\Theta}^{(k)}_i(S(\omega),t)\right|=\int_{I_{\text{res}}}cS(\omega)^{\frac{1}{2}}\left|\imath \omega R_i^{(k)}\right|\ud\omega=c\,\t{DVI}_i^{(k)}.
        \label{eq:ath_approaches_dvi}
    \end{equation}
    If one only considers the relative ranking of the all-time-high frequency response of nodes in a given network, but not their absolute values, the overall scaling factor $c$ in (\ref{eq:ath_approaches_dvi}) does not play a role in the ranking and we finally arrive at
    \begin{equation}
        \lim_{T\rightarrow\infty}\sigma_{\mathrm{ATH}}(i)=\sigma_{\mathrm{DVI}}(i)
        \quad \text{for all }i\in\{1,\cdots,N\}.
        \label{eq:ranking_converge}
    \end{equation}
\end{proof}

\begin{remark}[Generalization of DVI and convergence time]
    The DVI defined in \eqref{eq:dvi} provides a measure to estimate the relative nodal risk from the power grid network resonances induced by a unpredictable perturbation signal and thus helps to identify the most vulnerable nodes in networks with arbitrary topologies exhibiting particularly strong resonant responses. The integration interval of the frequencies in DVI is chosen to be the band of eigenfrequencies $I_{\t{res}}$ for a specific network so that its irregular resonance patterns are covered, however it is not a must. In principle, any frequency range can be chosen for DVI to estimate the relative strength of the  all-time-high responses in the specific frequency range. However, one should note that the timescale for the ranking of all-time-high responses to converge to the ranking given by DVI (in \ref{eq:ath_approaches_dvi} and \ref{eq:ranking_converge}) depends on the chosen frequency range. For instance, the convergence time would be longer if lower frequencies are included in the integration interval of DVI.
\end{remark}

\section{Role of LRT in uncovering response patterns}
\label{sec:role_LRT}

In the previous section we elaborated how network-wide dynamic response patterns of power grid systems can be extracted from the explicit solution of nodal responses given by the LRT of the oscillator model (see Sec.~\ref{sec:LRT_models}). In this section we summarize and compare the role of LRT in revealing different categories of the dynamic response patterns, such as the patterns emerging on different timescales, and in responses to perturbations with different levels of randomness and magnitude.

\subsection{Transient vs. steady-state responses}
\label{subsec:role_steady_transient}

As discussed in Sec.~\ref{subsec:steady_patterns} and Sec.~\ref{subsec:transient_patterns}, power grid transmission networks exhibit distinctive spatiotemporal response patterns on different timescales. The transient spreading pattern (see Sec.~\ref{subsec:transient_patterns}) of an external perturbation signal in a normally operating power grid network appears close to the time of impact $t=0$. It is characterized by a set of points in time, at which the impact arrives at individual units in the network. To a large extent, the topological dependence of the arrival times can be captured by a \textit{topological factor} which arises from the leading-term approximation of the linear response. The steady-state response patterns, in contrast, emerge as $t\rightarrow\infty$ (see Sec.~\ref{subsec:steady_patterns}), where the nodal responses to a sinusoidal perturbation converge to sinusoidal oscillations as well, but with various amplitudes. Consequently the set of the nodal response amplitudes, a time-invariant but frequency-dependent feature of the oscillating responses, constitute the steady-state response patterns characterizing three frequency regimes. The time scale separating transient from steady-state regimes is not a universal constant but intricately depends on several factors including network size $N$ and network topology, damping constant $\alpha$, and the specific node location we are interested in within the network.

Both transient and steady-state response patterns have been revealed and characterized through asymptotic analyses of the explicit solution of the linear nodal responses (\ref{eq:LRT_oscillator_sinusoidal}). The solution depends explicitly on time while the dependence on the network topology is implicit, embedded in the eigenvalues and eigenvectors of the weighted graph Laplacian matrix $\mathcal{L}$. Through asymptotic analyses, either with respect to time $t$ or the perturbation frequency $\omega$, the lengthy solution of the linear nodal response is reduced to one term per eigenmode that dominates the asymptotic behaviour as the variable $t$ or $\omega$ approaches its corresponding limit (see \autoref{prop:localized_pattern} and \autoref{prop:transient_leading_term}). As the contribution of each eigenmode contains the ``overlap factor'' $v_i^{[\ell]}v_k^{[\ell]}$ of the perturbed node $k$ and the responding node $i$, the powers of the Laplacian eigenvalue $(\lambda^{[\ell]})^m$ that is involved in the dominating term in each eigenmode $\ell$ translates to elements of the power of the Laplacian matrix $(\mathcal{L}^m)_{ki}$ through the summation over all $N$ eigenmodes $\ell\in\{0,\cdots,N-1\}$. Furthermore, with the help of the result from graph theory that $(\mathcal{L}^m)_{ki}\equiv0$ if $m\in\mathbb{N}_0$ is smaller than $d(k,i)$, the shortest path distance between node $k$ and $i$, the dependence of the nodal linear responses on graph-theoretic distance emerges. In this way one obtains the asymptotic spatiotemporal response patterns that depends explicitly on distance. 
 
A major difference between the patterns we uncover in steady-state responses and in transient responses is, the asymptotic behaviors of the steady-state response patterns are exact, in the sense that the higher order terms are negligible at the observed limits of the perturbation frequency $\omega$ (see Sec.~\ref{subsec:steady_patterns}). Meanwhile, the contribution of the higher order terms are not negligible in the patterns in transient responses: the leading-term approximation of the response exhibits a diverging error as the perturbation spreads further and the arrival time grows larger. Numerical simulations show that the higher order terms accounts for about $10\%$ of the actual arrival time of perturbations [\citet{zhang_topological_2020}], which is significant in predicting the perturbation spreading behaviour in real-world power grid systems. However, by means of numerical techniques, we can still use the topological factor proposed in \autoref{remark:topological_factor} to estimate the contribution of higher order terms $O(2d+3)$ in the Taylor series (\ref{eq:leading_term_approximation}) in a specific network ensemble and give accurate predictions for the actual arrival times of the impact of a perturbation [\citet{zhang_topological_2020}]. Related recent works \citet{wolter2018quantifying,schroder2019dynamic} studied transient propagation of perturbations in networked systems consisting of one-dimensional dynamical units. One main finding is a similar scaling of the unit's state variables $x_i(t)$ (or their deviations from a base state) with time $t$ as $x_i(t)\sim t^d$ where $d$ is the shortest-path distance between perturbed node and the node $i$ the response in measured at.

\subsection{Responses to deterministic perturbations vs. responses to unpredictable perturbations}
\label{subsec:role_deterministic_stochastic}

The power grid response patterns we discuss in this article can be classified into two categories: the ones emerging in the deterministic responses to a given perturbation signal, such as the transient responses (Sec.~\ref{subsec:transient_patterns}) and the steady-state responses (Sec.~\ref{subsec:steady_patterns}) to a given signal, and the ones that estimate the cumulative impact of an unpredictable signal on a power grid network, such as the DVI measuring the nodal risk of network resonances (Sec.~\ref{subsec:dvi}). Both categories of power grid response patterns are discovered based on the explicit solution of the linear nodal responses (\ref{eq:LRT_oscillator_sinusoidal}) that derived from the LRT.

Looking more closely, one finds that the estimated patterns in the cumulative responses can be seen as a result built upon the deterministic steady-state linear responses with one more dimension, i.e. the specific properties of the perturbation signals. At its core, the steady-state network frequency responses $\dot{\bm{\Theta}}$ (\ref{eq:steady_frequency_response_sinusoidal}) to a single-frequency sinusoidal oscillation $D_k(t)$ at a given node $k$ is a mapping $f_{G,\bm{\theta}^*}:\mathbb{R}\rightarrow\mathbb{R}^N$ with $f_{G,\bm{\theta}^*}$ depending on the underlying network topology $G$ and network's base state $\bm{\theta}^*$ prior to perturbations. The fluctuating nature of the perturbation signal adds another dimension to the responses: the amplitude of $D_k$ is no longer a given constant $\varepsilon$, but becomes further dependent of the frequency $\omega$ through the PSD $S(\omega)$ of the signal, i.e. $\varepsilon(\omega)=S(\omega)^{\frac{1}{2}}$. In this way, the unknown or unpredictable temporally detailed features of the perturbations are integrated into the LRT framework for networked dynamical systems, which extends standard LRT statements and enables us to estimate features of network responses beyond the deterministic realm of systems driven with known signals. 

We emphasize that, due to the intrinsic irregularity of fluctuating perturbation signals, the errors of the estimates for the associated responses appear to be significantly higher than the ones for the deterministic responses [compare \citet{zhang_fluctuation-induced_2019} and \citet{zhang_vulnerability_2020}]. In a finite signal time series characterized by a PSD function, the contained frequencies are finite in number in the considered frequency interval (such as the resonance regime $I_{\t{res}}$), and apparently resonant to each other, which leads to a  deviation in the ranking of the all-time-high nodal responses to the ranking given by indices computed \textit{a priori} (such as the DVI discussed in Sec.~\ref{subsec:dvi}). Additionally, realistic perturbation signals do not follow exactly the characteristic PSD, such that randomness also exists in the amplitudes of the frequency components. Nevertheless, compared to the deterministic patterns which gives only \textit{a posteriori} information of network responses, estimates such as the DVI given by the extended LRT may provide a useful guiding tool for risk assessments in real-world power grid systems.

\subsection{Small responses vs. large responses}
\label{subsec:role_small_large}

As the name suggests, LRT provides the linear approximation of a system's response to a perturbation close to a fixed point of a networked dynamical system. Therefore the solution given by LRT intrinsically deviates from the actual system responses due to the neglected higher order terms in the system's collective nonlinear dynamics. As the system being driven further and further away from the fixed point, the responses typically increase and so do the estimation error of the LRT. However the range of validity of the LRT, as well as how its error grows with the perturbation, is usually nontrivial and system-dependent.

For the oscillator model of power grid networks, the error of the solution given by LRT (\ref{eq:LRT_oscillator_step} and \ref{eq:LRT_oscillator_sinusoidal}) follows the same trend and grows with an increasingly stronger perturbation signal. However, numerical evidence  shows that the error increases mildly with the magnitude of the perturbation until it blows up close to a bifurcation point [\citet{zhang_fluctuation-induced_2019}]. In the linearized dynamics of the oscillator model at a fixed point $\bm{\theta}^*$ (\ref{eq:response_dynamics_oscillator_general}), the deviation of the nonlinear coupling terms $\sin(\theta_j-\theta_i)$ for all edges $(i,j)\in E$ to their values $\sin(\theta^*_j-\theta^*_i)$ at the fixed point are represented by the first-order approximations $\cos(\theta^*_j-\theta^*_i)(\Theta_j-\Theta_i)$, vectorized as the term $-\mathcal{L}\bm{\Theta}$ in (\ref{eq:response_dynamics_oscillator_general}). For power grid systems working at a stable operation state without any transmission line  overloaded ($|\theta_j^*-\theta_i^*|\leq\frac{\pi}{2}$ for all edges $(i,j)\in E$, see \autoref{prop:stability}), the linear approximation breaks down only when the system is driven far enough from the fixed point $\bm{\theta}^*$ and goes close to the point where one of the lines $(i,j)$ is fully loaded, i.e. $\sin(\theta_j-\theta_i)=1$. In this regime, the linear approximation diverges from the actual state of the system and the error grows explosively. As elaborated in an article by \citeasnoun{manik_supply_2014}, when one of the transmission line is fully loaded, the power grid system reaches a bifurcation point where the initial stable fixed point is lost and the oscillating units in the system becomes desynchronized. Therefore, the LRT of the oscillator model of the power grid systems, together with all of the derived response patterns, are generally valid as long as none of the lines become overloaded and the entire system becomes unstable [see \citet{zhang_fluctuation-induced_2019} for quantitative results of the LRT errors].

\section{Conclusions and outlook}
\label{sec:outlook}

\subsection{Conclusions and discussions}
In this work we systematically discuss how linear response theory (LRT) may shed light on the spatio-temporal response patterns emerging in networked dynamical systems under time-dependent perturbations. We exemplify a full analysis for model dynamics of power grid systems which are inevitably exposed to fluctuating power injections from renewable energy sources. Beyond previous works, we integrate and present all details required for a full mathematical analysis, specifically demonstrate how to evaluate the generally intricate, multiple-sum expressions determining spatio-temporal response patterns in a useful way and highlight how different results interconnect, for instance between transient and long-term dynamics or between different types of perturbations and across topologies. We introduce the main ideas of LRT and its general requirements for applicability on system settings (cf. Sec.~\ref{sec:general}). We explicate various aspects of application to models of power grid systems, such as i) the solution of linear responses of the stationary DC power flow model and of the dynamic oscillator model of AC power grids (cf. Sec.~\ref{sec:LRT_models}), and ii) how it helps to identify dynamic patterns in network-wide responses which provide theoretical guidelines for power grid design, control and risk assessments.

Although LRT has been widely used as a powerful tool in analyzing various response dynamics of many complex networked dynamical systems, the works presented in this article provide a fresh methodological angle to approach the problem. For power grid systems, LRT has been used to estimate e.g. quadratic performance measures for the network's overall excursion away from synchrony [\citet{tyloo_robustness_2018}, \citet{tyloo_key_2019}, \citet{plietzsch_bounds_2019}, \citet{coletta_transient_2018}] and the variance of the frequency response increment distribution [\citet{haehne_propagation_2019}]. For connectome dynamics in brain, LRT has also been used to analytically estimate the covariance of Gaussian linear model of the stochastically perturbed system [\citet{tononi_measure_1994}], which links the structural and the functional connectivities between brain regions [\citet{zamora-lopez_cortical_2010}, \citet{wang_hierarchical_2019}].

The work presented in this article approaches the dynamic network responses in a way different from the above-mentioned works: instead of quantifying the stochastic features of the overall or distance-specific responses directly based on the linear responses, we start from explicating the deterministic solution of network-wide responses to a single-frequency signal and using methods from graph theory and asymptotic analysis to extract spatiotemporal response patterns. These may be interpreted in a physically intuitive way. Specifically, the three frequency regimes of steady-state response patterns (Sec.~\ref{subsec:steady_patterns}) and the master curve of transient perturbation spreading (Sec.~\ref{subsec:transient_patterns}) are entirely deterministic. Especially, the former work on fluctuation-induced network resonances can be seen as a direct generalization of the classical resonance phenomenon of a single driven damped harmonic oscillator to oscillators interacting on networks. The emerging pattern constituted by the estimations of the all-time-high nodal response magnitudes to a irregularly varying perturbation (Sec.~\ref{subsec:dvi}), is also a straightforward result derived from the spatial patterns in the network-wide responses to a deterministic (periodic) signal.

\subsection{Challenges and future work}
Future work regarding the response theory for networked dynamical systems may follow several directions.

First, exact analytical solutions, and even many asymptotic results of the linear responses of general networked dynamical systems, as well as the response patterns emerging from these solutions, remain unknown to date. Applying the LRT presented in this work on a networked dynamical system, we employed several conditions on the system's dynamics (see Sec.~\ref{sec:general}) to explicate a full analysis without too many notational and other complications. The conditions include homogeneous nodal dynamics, a diffusive coupling term $g_{ij}(x_j-x_i)$, and the evenness of the coupling function's sensitivity $\frac{\ud g_{ij}}{\ud(x_j-x_i)}$ to small changes in the difference of nodal states $x_j-x_i$, such that a symmetric weighted Laplacian matrix arises in the linearized response dynamics of the system at the fixed (operating) point (\ref{eq:response_dynamics_Laplacian}). The presence of a symmetric Laplacian matrix in the linearized dynamics ensures the option to express the linear responses in the Laplacian eigenbasis, which plays a critical role in linking the response at a specific node to the graph-theoretic distance to the perturbation. Thereby it is critical also in uncovering the topological structure of the dynamical response patterns across the network. However, for many networked dynamical systems, such as the third-order model of power grid dynamics including the voltage dynamics [\citet{machowski_power_2008}], such preconditions are not fulfilled. One way to overcome this theoretical barrier and to extend LRT to such networks of dynamical systems is to transform the system's state variables to another coordinate system where the Jacobian matrix $\mathcal{J}$ in (\ref{eq:LRT_general}) is diagonal or almost diagonal (such as in the Jordan normal forms of $\mathcal{J}$). In this way explicit solutions for linear responses can be obtained in the new coordinate system where dynamic response patterns can be identified in similar ways presented in this work. 

Second, one could use LRT to develop strategies to control the impact of fluctuations on networked dynamical systems such as power grids. So far, we gained insights into the spatiotemporal structure of the responses across networks and developed indices to estimate the nodal risks against external  fluctuations. The next step towards more reliable and more robust power grid systems would be to utilize the obtained understanding to develop countermeasures against the risks, e.g. to suppress the potentially dangerous responses such as network resonances and to slow down the spreading of the impact of a sudden drop of injected power. A potential way to achieve such tasks could be to manipulate discovered response patterns by changing the interaction structure of the power grid.

Third, research on how different classes of network topologies potentially impact response patterns through their specific characteristics of their eigensystems. Progress in this direction seems hard, because one would need to be able to characterize, e.g., eigenvectors of graph ensembles such that they directly help to extract useful information from complex expressions like \eqref{eq:oscillator_response_factor} or even \eqref{eq:LRT_oscillator_sinusoidal} specifically for that ensemble.

Fourth, the LRT per se could be extended by considering also the higher-order approximations of the system's responses, cf. [\citet{Thuemler2022abstract}]. In the current work, we demonstrated that many features of collective response patterns of networked dynamical systems with nonlinear couplings, such as the oscillator model of power grids, are dominantly captured by the first-order (i.e. linear) approximation of the system's responses, yet nonlinear effects may also play a role for other systems with certain forms of the intrinsic nodal dynamics or the interaction dynamics between nodes, specifically if we ask for the loss of solutions near operating states [\citet{Thuemler2022abstract}]. Therefore it would be desirable if the contributions of the higher-order approximations of the system's responses can be estimated. An open question also here is how such nonlinear effects depend on the interaction topology of the network.

We conjecture that general network dynamical systems, also beyond power grids, similarly respond in characteristic ways to external input signals, making the systems non-equilibrium and often non-stationary, and to be described by non-autonomous deterministic or stochastic evolution equations. Several of the analysis steps presented above hint that the key methodological tools are either readily transferable to more general systems' settings or may be adapted to such settings. Candidate classes of systems include networks of multi-dimensional units, with discrete or hybrid dynamics, with delayed interactions or with spatially or temporally correlated stochastic inputs. Application areas may range from gene regulatory networks and metabolic circuits in cell biology to the controlled self-organized dynamics of engineered systems with feedback, from complex mechatronic systems to swarms of autonomous aerial vehicles.

\begin{acknowledgements}
    The authors acknowledge the support from the Deutsche Forschungsgemeinschaft (DFG; German Research Foundation) under Germany’s Excellence Strategy EXC-2068-390729961- Cluster of Excellence Physics of Life, the Center for Advancing Electronics at TU Dresden, by the German Federal Ministry for Research and Education (BMBF grants no. 03SF0472F and no. 03EK3055F), the National Natural Science Foundation of China (grant no. 12161141016), Shanghai Municipal Science and Technology Major Project (grant no. 2021SHZDZX0100), and the Shanghai Municipal Commission of Science and Technology Project (grant nos. 18ZR1442000 and 19511132101).
\end{acknowledgements}

\appendix
\section{Proof of \autoref{eq:polynomial_Q}}
\label{sec:proof_Q}
\begin{proposition}
    Given $\omega>0,\alpha>0$ and $0=\lambda^{[0]}<\cdots<\lambda^{[N-1]}$ as defined in \autoref{subsec:LRT_oscillator_model}, the product $Q^{[\ell]}(\omega)$
    \begin{equation}
        Q^{[\ell]}(\omega):=\prod_{\substack{\ell'=0,\ell'\neq \ell}}^{N-1}\left[\left(-\omega^2+\lambda^{[\ell']}\right)^2+\alpha^2\omega^2\right]
    \end{equation}
    that appears in the numerators of the real part and of the imaginary part of the nodal response strength (\ref{eq:nodal_response_strength}) explicitly depends on $\lambda^{[\ell]}$ and can be expressed as
    \begin{equation}
        Q(\lambda^{[\ell]},\omega)=\sum_{j=0}^{2N-2}C^{[j]}(\lambda^{[\ell]})\omega^{4N-4-2j}
    \end{equation}
    with the coefficient $C^{[j]}(\lambda^{[\ell]})$ is a polynomial of $\lambda^{[\ell]}$ with degree $j$.
    \label{prop:Q}
\end{proposition}
\begin{proof}
    To prove the proposition, we first rewrite the factors in $Q^{[\ell]}(w)$ by ordering the terms according to the degree of $\omega$:
    \begin{equation}
        Q^{[\ell]}(w)=\prod_{\substack{\ell'=0,\ell'\neq \ell}}^{N-1}\left[\omega^4+\left(\alpha^2-2\lambda^{[\ell']}\right)\omega^2+\left(\lambda^{[\ell']}\right)^2\right]=:\prod_{\substack{\ell'=0,\ell'\neq \ell}}^{N-1}\sum_{m=1}^3r_m(\lambda^{[\ell']})\omega^{2(3-m)}.
        \label{eq:Q_expanded}
    \end{equation}
    According to the distributivity of multiplication over addition, it is clear from (\ref{eq:Q_expanded}) that each term in $Q^{[\ell]}(w)$, a polynomial of $\omega$ with degree $4N-4$, can be seen as the product of three factors: $\prod_{\ell'\in {s_1}}r_1(\lambda^{[\ell']})\omega^4$, $\prod_{\ell'\in {s_2}}r_2(\lambda^{[\ell']})\omega^2$ and $\prod_{\ell'\in {s_3}}r_3(\lambda^{[\ell']})$, where sets $s_1$, $s_2$, and $s_3$ have $a:=|s_1|$, $b:=|s_2|$ and $c:=|s_3|$ elements respectively and together form a partition $P_{\ell}(a,b,c)$ of the set of the indices of the $N-1$ eigenmodes $S_{\ell}:=\{0,...,N-1\}\backslash\{\ell\}$. Here $a,b,c\in \mathbb{N}_0$ and satisfy $a+b+c=N-1$. Since $r_1$, $r_2$ and $r_3$ are polynomials of $\lambda^{[\ell']}$ with degree $0$, $1$ and $2$, a term in $Q^{[\ell]}(w)$ with $\omega^{4a+2b}$ would have a coefficient involving a multiplication of $2b+c$ Laplacian eigenvalues $\lambda^{[\ell']}$ with $\ell'\in S_{\ell}$. Denoting $j=b+2c$, we can write the coefficient of the term with degree $4a+2b=4N-4-2j$ as

	\begin{equation}
		C^{[\ell]}_j=\sum_{\substack{a+b+c=N-1\\b+2c=j}}\sum_{P_{\ell}(a,b,c)}\prod_{p\in s_2}\left(\alpha^2-2\lambda^{[p]}\right)\prod_{q\in s_3}\left(\lambda^{[q]}\right)^2,
		\label{eq:C_j}
	\end{equation}

    which is a sum over all possible partitions $P_{\ell}(a,b,c)$ satisfying $a+b+c=N-1$ and $b+2c=j$.
    
    In the following we show that the coefficient $C^{[\ell]}_j$ is a polynomial of $\lambda^{[\ell]}$ with degree $j$, i.e. $\deg[C^{[\ell]}_j(\lambda^{[\ell]})]=j$. For convenience of notation in the proof, we define the sum of coefficients involving both $r_2$ and $r_3$ over $s_2\in\binom{S_{\ell}}{b}$ and $s_3\in\binom{S_{\ell}\backslash s_2}{c}$, i.e. all possible partitions $P_{\ell}(a,b,c)$ of $S_{\ell}$ as

	\begin{equation}
		Y^{[\ell]}_{b,c}:=\sum_{s_2\in\binom{S_{\ell}}{b},s_3\in\binom{S_{\ell}\backslash s_2}{c}}\prod_{p\in s_2}\left(\alpha^2-2\lambda^{[p]}\right)\prod_{q\in s_3}\left(\lambda^{[q]}\right)^2.
        \label{eq:Yl_bc}
    \end{equation}
    Here $\binom{S_{\ell}}{b}$ denotes all possible $b$-subsets of $S_{\ell}$. Similarly, we define the sum of the coefficients over all possible partitions of $S:=\{0,...,N-1\}$ as

    \begin{equation}
        Y_{b,c}:=\sum_{s_2\in\binom{S}{b},s_3\in\binom{S\backslash s_2}{c}}\prod_{p\in s_2}\left(\alpha^2-2\lambda^{[p]}\right)\prod_{q\in s_3}\left(\lambda^{[q]}\right)^2.
        \label{eq:Y_bc}
    \end{equation}
    In case $b=0$ or $c=0$, the corresponding product is omitted. It is clear that $Y_{b,c}$, including special cases $Y_{b,0}$ and $Y_{0,c}$ are constants independent of $\lambda^{[\ell]}$. Using definition (\ref{eq:Yl_bc}) and (\ref{eq:Y_bc}), we can write $C_j^{[\ell]}$ in \ref{eq:C_j} as $\sum_{\substack{a+b+c=N-1,b+2c=j}}Y^{[\ell]}_{b,c}$. To prove \autoref{prop:Q}, we only need to prove $Y^{[\ell]}_{b,c}$ is a polynomial of $\lambda^{[\ell]}$ with degree $j=b+2c$, i.e.
    \begin{equation}
        \deg\left[Y^{[\ell]}_{b,c}(\lambda^{[\ell]})\right]=b+2c.
        \label{eq:Yl_bc_polynomial}
    \end{equation}
    Now we show (\ref{eq:Yl_bc_polynomial}) in three steps. All subproofs are given by mathematical induction.
    
    \begin{enumerate}[\hspace{4pt}(P1)]
        \item[Step 1:] First, we show that the sum of $r_2$-related factors over $s_2\in\binom{S_{\ell}}{b}$ is a polynomial of $\lambda^{[\ell]}$ with degree $b$. That is, $\deg[Y^{[\ell]}_{(b,0)}(\lambda^{[\ell]})]=b$.
        \begin{enumerate}[\hspace{-12pt}(P1)]
            \item For $b=1$, we have $Y^{[\ell]}_{1,0}=Y_{1,0}-(\alpha^2-2\lambda^{[\ell]})$, which is a polynomial of $\lambda^{[\ell]}$ with degree $1$ since $Y_{1,0}$ is a constant independent of $\lambda^{[\ell]}$.
            \item If the statement holds for $b=n-1$, i.e. $\deg[Y^{[\ell]}_{n-1,0}(\lambda^{[\ell]})]=n-1$, then for $b=n$ we have $Y^{[\ell]}_{n,0}=Y_{n,0}-\left(\alpha^2-2\lambda^{[\ell]}\right)Y^{[\ell]}_{n-1,0}$, satisfying $\deg[Y^{[\ell]}_{n,0}(\lambda^{[\ell]})]=n$.
        \end{enumerate}
        
        \item[Step 2:] Second, we show that the sum of $r_3$-related factors over $s_3\in\binom{S_{\ell}}{c}$ is a polynomial of $\lambda^{[\ell]}$ with degree $2c$. That is, $\deg[Y^{[\ell]}_{(0,c)}(\lambda^{[\ell]})]=2c$.
        \begin{enumerate}[\hspace{-12pt}(P1)]
            \item For $c=1$, we have $Y^{[\ell]}_{0,1}=Y_{0,1}-(\lambda^{[\ell]})^2$, which is a polynomial of $\lambda^{[\ell]}$ with degree $2$ since $Y_{0,1}$ is a constant independent of $\lambda^{[\ell]}$.
            \item If the statement holds for $c=n-1$, i.e. $\deg[Y^{[\ell]}_{0,n-1}(\lambda^{[\ell]})]=2n-2$, then for $c=n$ we have $Y^{[\ell]}_{0,n}=Y_{0,n}-\left(\lambda^{[\ell]}\right)^2Y^{[\ell]}_{0,n-1}$, satisfying $\deg[Y^{[\ell]}_{0,n}(\lambda^{[\ell]})]=2n$.
        \end{enumerate}
        
        \item[Step 3:] Finally, we show that the sum of coefficients involving both $r_2$ and $r_3$ over $s_2\in\binom{S_{\ell}}{b}$ and $s_3\in\binom{S_{\ell}\backslash\{s_2\}}{c}$, i.e. all possible partitions $P_{\ell}(a,b,c)$, is a polynomial of $\lambda^{[\ell]}$ with degree $j=b+2c$. That is, equation (\ref{eq:Yl_bc_polynomial}).
        \begin{enumerate}[\hspace{-12pt}(P1)]
            \item For $b=1,c=1$, we have $Y^{[\ell]}_{1,1}=Y_{1,1}-(\alpha^2-2\lambda^{[\ell]})Y^{[\ell]}_{0,1}-(\lambda^{[\ell]})^2Y^{[\ell]}_{1,0}$, which is a polynomial of $\lambda^{[\ell]}$ with degree $3$ since $\deg[Y_{(1,1)}(\lambda^{[\ell]})]=0$, $\deg[Y^{[\ell]}_{(0,1)}(\lambda^{[\ell]})]=2$ and $\deg[Y^{[\ell]}_{(1,0)}(\lambda^{[\ell]})]=1$.
            \item If the statement holds for $b=m-1,c=n-1$, i.e. $\deg[Y^{[\ell]}_{m-1,n-1}(\lambda^{[\ell]})]=m+2n-3$, then for $b=m,c=n$ we have
            \begin{align}
                Y^{[\ell]}_{m,n}=&Y_{m,n}-\left(\alpha^2-2\lambda^{[\ell]}\right)Y^{[\ell]}_{m-1,n}-\left(\lambda^{[\ell]}\right)^2Y^{[\ell]}_{m,n-1}\nonumber\\
                =&Y_{m,n}-\left(\alpha^2-2\lambda^{[\ell]}\right)\left(Y_{m-1,n}-\left(\lambda^{[\ell]}\right)^2 Y^{[\ell]}_{m-1,n-1}\right)\nonumber\\
	            & -\left(\lambda^{[\ell]}\right)^2 \left(Y_{m,n-1}-\left(\alpha^2-2\lambda^{[\ell]}\right) Y^{[\ell]}_{m-1,n-1}\right).\nonumber
            \end{align}
            Taking into account that $Y_{m,n}$, $Y_{m-1,n}$ and $Y_{m,n-1}$ all have degree $0$, we can easily see that $\deg[Y^{[\ell]}_{m,n}(\lambda^{[\ell]})]=m+2n$, meaning the statement also holds for $b=m$ and $c=n$.
        \end{enumerate}
    \end{enumerate}
    
\end{proof}

\section{Proof of \autoref{eq:derivative_leading_term}}
\label{sec:proof_F}
\begin{proposition}
    The function 
    \begin{equation}
        F_n(\lambda^{[\ell]}):=\dfrac{1}{-\omega^2+\imath\alpha\omega+\lambda^{[\ell]}}\left[\dfrac{\left(\Delta^{[\ell]}_+\right)^n\left(\Delta^{[\ell]}_--\imath\omega\right)-\left(\Delta^{[\ell]}_-\right)^n\left(\Delta^{[\ell]}_+-\imath\omega\right)}{2\eta^{[\ell]}}+(\imath\omega)^n\right]
    \end{equation}
    that appears in the $n$-th order derivative of the linear response at $t=0$ (\ref{eq:transient_derivative}) has a leading term with respect to $\lambda^{[\ell]}$
    \begin{align}
        \lt{F_n(\lambda^{[\ell]})}=\left\lbrace
		  \begin{array}{ll}
			(-1)^{\frac{n-1}{2}}\left(-\imath\omega+\frac{n-1}{2}\alpha\right)\left(\lambda^{[\ell]}\right)^{\frac{n-3}{2}}& \quad \text{if } n \text{ is odd,}\\
		   \left(-\lambda^{[\ell]}\right)^{\frac{n-2}{2}}& \quad \text{if } n \text{ is even.}
		  \end{array}\right.
		 \label{eq:derivative_leading_term_appendix}
    \end{align}
    Here $\Delta^{[\ell]}_{\pm}:=-\alpha/2\pm\eta^{[\ell]}$, $\eta^{[\ell]}:=\sqrt{\alpha^2/4-\lambda^{[\ell]}}$ with $\alpha>0$, $0=\lambda^{[0]}<\cdots<\lambda^{[N-1]}$ and $\omega>0$, $n\in\mathbb{N}$, $n\geq2$.
\end{proposition}

\begin{proof}
    Using the relation $\Delta^{[\ell]}_+\Delta^{[\ell]}_-=\lambda^{[\ell]}$ we rewrite the function under study as
    \begin{equation}
        F_n(\lambda^{[\ell]})=\dfrac{\lambda^{[\ell]}f_{n-1}(\lambda^{[\ell]})-\imath\omega f_n(\lambda^{[\ell]})+(\imath\omega)^n}{-\omega^2+\imath\alpha\omega+\lambda^{[\ell]}}
        \label{eq:polynomial_F_f}
    \end{equation}
    with
    \begin{equation}
        f_n(\lambda^{[\ell]}):=\dfrac{1}{2\eta^{[\ell]}}\left[(\Delta^{[\ell]}_+)^n-(\Delta^{[\ell]}_-)^n\right].
    \end{equation}
    It is clear from (\ref{eq:polynomial_F_f}) that the leading term of $F_n(\lambda^{[\ell]})$ depends on the leading term of $f_n(\lambda^{[\ell]})$ as
    \begin{align}
        \lt{F_n(\lambda^{[\ell]})}&=\frac{\lt{\lambda^{[\ell]}f_{n-1}(\lambda^{[\ell]})-\imath\omega f_n(\lambda^{[\ell]})+(\imath\omega)^n}}{\lt{-\omega^2+\imath\alpha\omega+\lambda^{[\ell]}}}\\
        &=\frac{1}{\lambda^{[\ell]}}\lt{\lambda^{[\ell]}f_{n-1}(\lambda^{[\ell]})-\imath\omega f_n(\lambda^{[\ell]})}.
        \label{eq:lt_F_f}
    \end{align} 
    
    Please note that $(\Delta^{[\ell]}_+)^n$ and $(\Delta^{[\ell]}_-)^n$ in $f_n(\lambda^{[\ell]})$ are a complex conjugate pair, since $\eta^{[\ell]}$ is imaginary under the low dissipation of power grid systems (see \autoref{remark:oscillator_eigenfrequencies}). Therefore we have $(\Delta^{[\ell]}_-)^{n}=\overline{(\Delta^{[\ell]}_+)^{n}}$ which leads to $f_n(\lambda^{[\ell]})=\text{Im}\,(\Delta^{[\ell]}_+)^{n}/\sqrt{\lambda^{[\ell]}-\alpha^2/4}$. Now proving (\ref{eq:derivative_leading_term_appendix}) boils down to determining the leading term of $(\Delta^{[\ell]}_+)^n$. In the following we use mathematical induction to show that leading term of the real part and the imaginary part of $(\Delta^{[\ell]}_+)^n$ follows
    \begin{align}
        \lt{\text{Re}\,(\Delta^{[\ell]}_+)^n}&=
	    \left\lbrace
	    \begin{array}{ll}
	        (-1)^{\frac{n+1}{2}}\frac{n}{2}\alpha\left(\lambda^{[\ell]}\right)^{\frac{n-1}{2}}      & \quad \text{if } n \text{ is odd,}\\
	        \left(-\lambda^{[\ell]}\right)^{\frac{n}{2}}     & \quad \text{if } n \text{ is even;}\\
	    \end{array}\right.\label{eq:lt_delta_re}\\
	    \lt{f_n(\lambda^{[\ell]})}=\lt{\dfrac{\text{Im}\,(\Delta^{[\ell]}_+)^n}{\sqrt{\lambda^{[\ell]}-\alpha^2/4}}}&=
    	\left\lbrace
	    \begin{array}{ll}
	        \left(-\lambda^{[\ell]}\right)^{\frac{n-1}{2}}      & \quad \text{if } n \text{ is odd,}\\
	        (-1)^{\frac{n}{2}}\frac{n}{2}\alpha \left(\lambda^{[\ell]}\right)^{\frac{n-2}{2}}     & \quad \text{if } n \text{ is even.}\\
	    \end{array}\right.\label{eq:lt_fn}
    \end{align}
    \begin{enumerate}
        \item[(a)] For $n=2$ and $n=3$, we can easily verify (\ref{eq:lt_delta_re}) and (\ref{eq:lt_fn}) by spelling out $(\Delta^{[\ell]}_+)^n$:
        \begin{align}
            &\lt{\text{Re}\,(\Delta^{[\ell]}_+)^2}=\lt{-\lambda^{[\ell]}+\tfrac{1}{2}\alpha^2}=-\lambda^{[\ell]},
            \nonumber\\
            &\lt{\text{Re}\,(\Delta^{[\ell]}_+)^3}=\lt{\tfrac{3}{2}\alpha\lambda^{[\ell]}-\tfrac{1}{4}\alpha^2-\tfrac{1}{4}\alpha^3}=\tfrac{3}{2}\alpha\lambda^{[\ell]},\nonumber\\
            &\lt{f_2(\lambda^{[\ell]})}=\lt{-\alpha}=-\alpha,\quad\lt{f_3(\lambda^{[\ell]})}=\lt{-\lambda^{[\ell]}+\alpha^2}=-\lambda^{[\ell]}\nonumber.
        \end{align}
        \item[(b)] Now we show that (\ref{eq:lt_delta_re}) and (\ref{eq:lt_fn}) hold for $n+1$ if they hold for $n$, no matter $n$ is odd or even. The leading term of $\Real\,(\Delta^{[\ell]}_+)^{n+1}$ and $f_{n+1}(\lambda^{[\ell]})$ can be expressed in terms of the leading term of $\Real\,(\Delta^{[\ell]}_+)^{n}$ and $f_{n}(\lambda^{[\ell]})$ as following
        \begin{align}
            \lt{\Real\,(\Delta^{[\ell]}_+)^{n+1}}&=\lt{(-\tfrac{1}{2}\alpha)\lt{\Real\,(\Delta^{[\ell]}_+)^{n}}-\lt{f_n(\lambda^{[\ell]})}(\lambda^{[\ell]}-\tfrac{1}{4}\alpha^2)},\\
            \lt{f_{n+1}(\lambda^{[\ell]})}&=\lt{\lt{\Real\,(\Delta^{[\ell]}_+)^{n}}+(-\tfrac{1}{2}\alpha)\lt{f_{n}(\lambda^{[\ell]})}}.
        \end{align}
        In case $n$ is odd, we have
        \begin{align}
            \lt{\Real\,(\Delta^{[\ell]}_+)^{n+1}}&=\lt{(-1)^{\frac{n+3}{2}}\tfrac{n}{4}\alpha^2(\lambda^{[\ell]})^{\frac{n-1}{2}}+(-\lambda^{[\ell]})^{\frac{n-1}{2}}(\lambda^{[\ell]}-\tfrac{1}{4}\alpha^2)}\nonumber\\
            &=(-\lambda^{[\ell]})^{\frac{n+1}{2}},\text{ and}\label{eq:n_odd_re}\\
            \lt{f_{n+1}(\lambda^{[\ell]})}&=\lt{ (-1)^{\frac{n+1}{2}}\tfrac{n}{2}\alpha(\lambda^{[\ell]})^{\frac{n-1}{2}}-\tfrac{1}{2}\alpha(-\lambda^{[\ell]})^{\frac{n-1}{2}}}\nonumber\\
            &=(-1)^{\frac{n+1}{2}}\tfrac{n+1}{2}\alpha(\lambda^{[\ell]})^{\frac{n-1}{2}}\label{eq:n_odd_f}
        \end{align}
        In case $n$ is even, we have
        \begin{align}
            \lt{\Real\,(\Delta^{[\ell]}_+)^{n+1}}&=\lt{(-\lambda^{[\ell]})^{\frac{n}{2}}(-\tfrac{1}{2}\alpha)+(-1)^{\frac{n}{2}}\dfrac{n}{2}\alpha (\lambda^{[\ell]})^{\frac{n-2}{2}}(\lambda^{[\ell]}-\tfrac{1}{4}\alpha^2)}\nonumber\\
            &=(-1)^{\frac{n+2}{2}}\tfrac{n+1}{2}\alpha(\lambda^{[\ell]})^{\frac{n}{2}},\text{ and}\label{eq:n_even_re}\\
            \lt{f_{n+1}(\lambda^{[\ell]})}&=\lt{(-\lambda^{[\ell]})^{\frac{n}{2}}+(-1)^{\frac{n+2}{2}}\tfrac{n+1}{2}\alpha (\lambda^{[\ell]})^{\frac{n-2}{2}}}=(-\lambda^{[\ell]})^{\frac{n}{2}}.\label{eq:n_even_f}
        \end{align}
        The results (\ref{eq:n_odd_re}, \ref{eq:n_odd_f}, \ref{eq:n_even_re}, \ref{eq:n_even_f}) agree with the statement (\ref{eq:lt_delta_re}) and (\ref{eq:lt_fn}).
    \end{enumerate}
    Combining results (\ref{eq:lt_F_f}) and (\ref{eq:lt_fn}), we arrive at the leading term of $F_n(\lambda^{[\ell]})$ as
    \begin{align}
        \lt{F_n(\lambda^{[\ell]})}&=\dfrac{1}{\lambda^{[\ell]})}\lt{(-1)^{\frac{n-1}{2}}\tfrac{n-1}{2}\alpha(\lambda^{[\ell]})^{\frac{n-1}{2}}-\imath\omega (-\lambda^{[\ell]})^{\frac{n-1}{2}}}\nonumber\\
        &=(-1)^{\frac{n-1}{2}}\left(-\imath\omega+\tfrac{n-1}{2}\alpha\right)(\lambda^{[\ell]})^{\frac{n-3}{2}} \quad \text{if } n \text{ is odd, and}\\
        \lt{F_n(\lambda^{[\ell]})}&=\dfrac{1}{\lambda^{[\ell]})}\lt{(-1)^{\frac{n-2}{2}}(\lambda^{[\ell]})^{\frac{n}{2}}-\imath\omega(-1)^{\frac{n}{2}}\tfrac{n}{2}\alpha(\lambda^{[\ell]})^{\frac{n-2}{2}}}\nonumber\\
        &=(-\lambda^{[\ell]})^{\frac{n-2}{2}}\quad \text{if } n \text{ is even.}
    \end{align}
\end{proof}

\bibliographystyle{dcu}
\bibliography{FluctuationResponsePatterns}

@article{Thuemler2022abstract,
	author = {Th\"umler, M. and Schr\"oder, M. and Timme, M.},
	journal = {IFAC conference abstract},
	title = {Nonlinear and divergent responses of fluctuation-driven systems},
	year = {2022}}

@article{schultz_random_2014,
	author = {Schultz, Paul and Heitzig, Jobst and Kurths, J{\"u}rgen},
	doi = {10.1140/epjst/e2014-02279-6},
	issn = {19516401},
	journal = {European Physical Journal: Special Topics},
	title = {A random growth model for power grids and other spatially embedded infrastructure networks},
	year = {2014}}

@article{tyloo_key_2019,
	author = {Tyloo, M. and Pagnier, L. and Jacquod, P.},
	doi = {10.1126/sciadv.aaw8359},
	issn = {2375-2548},
	journal = {Science Advances},
	language = {en},
	month = nov,
	number = {11},
	pages = {eaaw8359},
	shorttitle = {The key player problem in complex oscillator networks and electric power grids},
	title = {The key player problem in complex oscillator networks and electric power grids: {Resistance} centralities identify local vulnerabilities},
	urldate = {2021-07-27},
	volume = {5},
	year = {2019}}

@article{dorfler_synchronization_2013,
	author = {D\"orfler, F. and Chertkov, M. and Bullo, F.},
	doi = {10.1073/pnas.1212134110},
	issn = {0027-8424, 1091-6490},
	journal = {Proceedings of the National Academy of Sciences},
	language = {en},
	month = feb,
	number = {6},
	pages = {2005--2010},
	title = {Synchronization in complex oscillator networks and smart grids},
	urldate = {2021-07-27},
	volume = {110},
	year = {2013}}

@article{manik_network_2017,
	author = {Manik, D. and Rohden, M. and Ronellenfitsch, H. and Zhang, X. and Hallerberg, S. and Witthaut, D. and Timme, M.},
	copyright = {All rights reserved},
	doi = {10.1103/PhysRevE.95.012319},
	issn = {24700053},
	journal = {Physical Review E},
	number = {1},
	title = {Network susceptibilities: {Theory} and applications},
	volume = {95},
	year = {2017}}

@article{kettemann_delocalization_2016,
	author = {Kettemann, Stefan},
	doi = {10.1103/PhysRevE.94.062311},
	issn = {24700053},
	journal = {Physical Review E},
	note = {\_eprint: 1504.05525},
	title = {Delocalization of disturbances and the stability of ac electricity grids},
	year = {2016}}

@article{haehne_propagation_2019,
	author = {Haehne, Hauke and Schmietendorf, Katrin and Tamrakar, Samyak and Peinke, Joachim and Kettemann, Stefan},
	doi = {10.1103/PhysRevE.99.050301},
	issn = {24700053},
	journal = {Physical Review E},
	note = {\_eprint: 1809.09098},
	title = {Propagation of wind-power-induced fluctuations in power grids},
	year = {2019}}

@article{tamrakar_propagation_2018,
	author = {Tamrakar, Samyak and Conrath, Michael and Kettemann, Stefan},
	doi = {10.1038/s41598-018-24685-5},
	issn = {20452322},
	journal = {Scientific Reports},
	note = {\_eprint: 1706.10144},
	title = {Propagation of disturbances in {AC} electricity grids},
	year = {2018}}

@book{bapat_graphs_2010,
	address = {London ; New York : New Delhi},
	author = {Bapat, R. B.},
	isbn = {978-1-84882-980-0 978-1-84882-981-7},
	note = {OCLC: ocn455828013},
	publisher = {Springer ; Hindustan Book Agency},
	series = {Universitext},
	title = {Graphs and matrices},
	year = {2010}}

@article{pan_non-hermitian_2020,
	author = {Pan, Lei and Chen, Xin and Chen, Yu and Zhai, Hui},
	doi = {10.1038/s41567-020-0889-6},
	issn = {1745-2473, 1745-2481},
	journal = {Nature Physics},
	language = {en},
	month = jul,
	number = {7},
	pages = {767--771},
	title = {Non-{Hermitian} linear response theory},
	urldate = {2021-07-15},
	volume = {16},
	year = {2020}}

@article{majda_linear_2010,
	author = {Majda, Andrew and Wang, Xiaoming},
	doi = {10.4310/CMS.2010.v8.n1.a8},
	issn = {15396746, 19450796},
	journal = {Communications in Mathematical Sciences},
	language = {en},
	number = {1},
	pages = {145--172},
	title = {Linear response theory for statistical ensembles in complex systems with time-periodic forcing},
	urldate = {2021-07-15},
	volume = {8},
	year = {2010}}

@article{ikeguchi_protein_2005,
	author = {Ikeguchi, Mitsunori and Ueno, Jiro and Sato, Miwa and Kidera, Akinori},
	doi = {10.1103/PhysRevLett.94.078102},
	issn = {0031-9007, 1079-7114},
	journal = {Physical Review Letters},
	language = {en},
	month = feb,
	number = {7},
	pages = {078102},
	shorttitle = {Protein {Structural} {Change} {Upon} {Ligand} {Binding}},
	title = {Protein {Structural} {Change} {Upon} {Ligand} {Binding}: {Linear} {Response} {Theory}},
	urldate = {2021-07-15},
	volume = {94},
	year = {2005}}

@article{cammi_linear_1999,
	author = {Cammi, Roberto and Mennucci, Benedetta},
	doi = {10.1063/1.478861},
	issn = {0021-9606, 1089-7690},
	journal = {The Journal of Chemical Physics},
	language = {en},
	month = may,
	number = {20},
	pages = {9877--9886},
	title = {Linear response theory for the polarizable continuum model},
	urldate = {2021-07-15},
	volume = {110},
	year = {1999}}

@article{ruelle_review_2009,
	author = {Ruelle, David},
	doi = {10.1088/0951-7715/22/4/009},
	issn = {0951-7715, 1361-6544},
	journal = {Nonlinearity},
	month = apr,
	number = {4},
	pages = {855--870},
	title = {A review of linear response theory for general differentiable dynamical systems},
	urldate = {2021-07-15},
	volume = {22},
	year = {2009}}

@article{kubo_statistical-mechanical_1957,
	author = {Kubo, Ryogo},
	doi = {10.1143/JPSJ.12.570},
	issn = {0031-9015, 1347-4073},
	journal = {Journal of the Physical Society of Japan},
	language = {en},
	month = jun,
	number = {6},
	pages = {570--586},
	title = {Statistical-{Mechanical} {Theory} of {Irreversible} {Processes}. {I}. {General} {Theory} and {Simple} {Applications} to {Magnetic} and {Conduction} {Problems}},
	urldate = {2021-07-15},
	volume = {12},
	year = {1957}}

@techreport{entso-e_continental_2021,
	author = {ENTSO-E},
	institution = {European Network of Transmission System Operators for Electricity},
	month = feb,
	note = {\url{https://eepublicdownloads.azureedge.net/clean-documents/Publications/Position\%20papers\%20and\%20reports/entso-e_CESysSep_interim_report_210225.pdf}},
	title = {Continental {Europe} {Synchronous} {Area} {Separation} on 8 {January} 2021: {Interim} {Report}},
	year = {2021}}

@techreport{investigation_committee_final_2007,
	author = {UCTE},
	institution = {Union for the Co-ordination of Transmission of Electricity},
	month = jan,
	note = {\url{https://eepublicdownloads.entsoe.eu/clean-documents/pre2015/publications/ce/otherreports/Final-Report-20070130.pdf}},
	title = {Final {Report} - {System} {Disturbance} on 4 {November} 2006},
	year = {2007}}

@article{manik_supply_2014,
	author = {Manik, Debsankha and Witthaut, Dirk and Sch{\"a}fer, Benjamin and Matthiae, Moritz and Sorge, Andreas and Rohden, Martin and Katifori, Eleni and Timme, Marc},
	doi = {10.1140/epjst/e2014-02274-y},
	issn = {1951-6355, 1951-6401},
	journal = {The European Physical Journal Special Topics},
	language = {en},
	month = oct,
	number = {12},
	pages = {2527--2547},
	shorttitle = {Supply networks},
	title = {Supply networks: {Instabilities} without overload},
	urldate = {2021-05-05},
	volume = {223},
	year = {2014}}

@article{zhang_vulnerability_2020,
	author = {Zhang, Xiaozhu and Ma, Cheng and Timme, Marc},
	copyright = {All rights reserved},
	doi = {10.1063/1.5122963},
	issn = {1054-1500, 1089-7682},
	journal = {Chaos: An Interdisciplinary Journal of Nonlinear Science},
	language = {en},
	month = jun,
	number = {6},
	pages = {063111},
	title = {Vulnerability in dynamically driven oscillatory networks and power grids},
	urldate = {2021-02-22},
	volume = {30},
	year = {2020}}

@article{anvari_short_2016,
	author = {Anvari, M and Lohmann, G and W{\"a}chter, M and Milan, P and Lorenz, E and Heinemann, D and Tabar, M Reza Rahimi and Peinke, Joachim},
	doi = {10.1088/1367-2630/18/6/063027},
	issn = {1367-2630},
	journal = {New Journal of Physics},
	month = jun,
	number = {6},
	pages = {063027},
	title = {Short term fluctuations of wind and solar power systems},
	urldate = {2021-06-26},
	volume = {18},
	year = {2016}}

@article{zhang_fluctuation-induced_2019,
	author = {Zhang, Xiaozhu and Hallerberg, Sarah and Matthiae, Moritz and Witthaut, Dirk and Timme, Marc},
	copyright = {All rights reserved},
	doi = {10.1126/sciadv.aav1027},
	issn = {23752548},
	journal = {Science Advances},
	number = {7},
	pages = {eaav1027},
	title = {Fluctuation-induced distributed resonances in oscillatory networks},
	volume = {5},
	year = {2019}}

@article{zhang_topological_2020,
	author = {Zhang, Xiaozhu and Witthaut, Dirk and Timme, Marc},
	copyright = {All rights reserved},
	doi = {10.1103/PhysRevLett.125.218301},
	issn = {0031-9007, 1079-7114},
	journal = {Physical Review Letters},
	language = {en},
	month = nov,
	number = {21},
	pages = {218301},
	title = {Topological {Determinants} of {Perturbation} {Spreading} in {Networks}},
	urldate = {2021-02-22},
	volume = {125},
	year = {2020}}

@book{kundur_power_1994,
	address = {New York},
	author = {Kundur, P. and Balu, Neal J. and Lauby, Mark G.},
	isbn = {978-0-07-035958-1},
	publisher = {McGraw-Hill},
	series = {The {EPRI} power system engineering series},
	title = {Power system stability and control},
	year = {1994}}

@article{motter_spontaneous_2013,
	author = {Motter, Adilson E. and Myers, Seth A. and Anghel, Marian and Nishikawa, Takashi},
	doi = {10.1038/nphys2535},
	issn = {17452481},
	journal = {Nature Physics},
	month = feb,
	note = {Publisher: Nature Publishing Group},
	number = {3},
	pages = {191--197},
	title = {Spontaneous synchrony in power-grid networks},
	volume = {9},
	year = {2013}}

@article{witthaut_critical_2016,
	author = {Witthaut, D. and Rohden, M. and Zhang, X. and Hallerberg, S. and Timme, M.},
	copyright = {All rights reserved},
	doi = {10.1103/PhysRevLett.116.138701},
	issn = {10797114},
	journal = {Physical Review Letters},
	number = {13},
	title = {Critical {Links} and {Nonlocal} {Rerouting} in {Complex} {Supply} {Networks}},
	volume = {116},
	year = {2016}}

@article{rohden_self-organized_2012,
	author = {Rohden, Martin and Sorge, Andreas and Timme, Marc and Witthaut, Dirk},
	doi = {10.1103/PhysRevLett.109.064101},
	issn = {00319007},
	journal = {Physical Review Letters},
	month = aug,
	note = {Publisher: American Physical Society},
	number = {6},
	pages = {064101},
	title = {Self-{Organized} {Synchronization} in {Decentralized} {Power} {Grids}},
	volume = {109},
	year = {2012}}

@book{machowski_power_2008,
	address = {Chichester, U.K},
	author = {Machowski, Jan and Bialek, Janusz W. and Bumby, J. R.},
	edition = {2nd ed},
	isbn = {978-0-470-72558-0},
	note = {OCLC: ocn232130756},
	publisher = {Wiley},
	shorttitle = {Power system dynamics},
	title = {Power system dynamics: stability and control},
	year = {2008}}

@book{milano_power_2010,
	address = {Berlin, Heidelberg},
	author = {Milano, Federico},
	doi = {10.1007/978-3-642-13669-6},
	isbn = {978-3-642-13668-9 978-3-642-13669-6},
	publisher = {Springer Berlin Heidelberg},
	series = {Power {Systems}},
	title = {Power {System} {Modelling} and {Scripting}},
	urldate = {2021-05-19},
	year = {2010}}

@techreport{wilde_ofgem_2020,
	author = {Wilde, Simon},
	institution = {Office of Gas and Electricity Markets UK},
	month = jan,
	note = {\url{https://www.ofgem.gov.uk/system/files/docs/2020/01/9_august_2019_power_outage_report.pdf}},
	title = {Ofgem 9 {August} 2019 power outage report},
	year = {2020}}

@article{filatrella_analysis_2008,
	author = {Filatrella, G. and Nielsen, A. H. and Pedersen, N. F.},
	doi = {10.1140/epjb/e2008-00098-8},
	issn = {1434-6028, 1434-6036},
	journal = {The European Physical Journal B},
	language = {en},
	month = feb,
	number = {4},
	pages = {485--491},
	title = {Analysis of a power grid using a {Kuramoto}-like model},
	urldate = {2021-05-07},
	volume = {61},
	year = {2008}}

@article{postnov_neural_2006,
	author = {Postnov, Dmitry E. and Ryazanova, Ludmila S. and Mosekilde, Erik and Sosnovtseva, Olga V.},
	doi = {10.1142/S0129065706000536},
	issn = {0129-0657, 1793-6462},
	journal = {International Journal of Neural Systems},
	language = {en},
	month = apr,
	number = {02},
	pages = {99--109},
	title = {{Neural} synchronization via Potassium signaling},
	urldate = {2021-05-07},
	volume = {16},
	year = {2006}}

@article{casagrande_synchronization_2006,
	author = {Casagrande, Vanessa},
	collaborator = {Technische Universit{\"a}t Berlin and Technische Universit{\"a}t Berlin and Engel, Harald},
	copyright = {Terms of German Copyright Law},
	doi = {10.14279/DEPOSITONCE-1345},
	month = apr,
	note = {Publisher: Technische Universit{\"a}t Berlin},
	title = {Synchronization, {waves}, and {turbulence} in {systems} of {interacting} {chemical} {oscillators}},
	urldate = {2021-05-07},
	year = {2006}}

@article{acebron_kuramoto_2005,
	author = {Acebr{\'o}n, Juan A. and Bonilla, L. L. and P{\'e}rez Vicente, Conrad J. and Ritort, F{\'e}lix and Spigler, Renato},
	doi = {10.1103/RevModPhys.77.137},
	issn = {0034-6861, 1539-0756},
	journal = {Reviews of Modern Physics},
	language = {en},
	month = apr,
	number = {1},
	pages = {137--185},
	shorttitle = {The {Kuramoto} model},
	title = {The {Kuramoto} model: {A} simple paradigm for synchronization phenomena},
	urldate = {2021-05-07},
	volume = {77},
	year = {2005}}

@article{larter_coupled_1999,
	author = {Larter, Raima and Speelman, Brent and Worth, Robert M.},
	doi = {10.1063/1.166453},
	issn = {1054-1500, 1089-7682},
	journal = {Chaos: An Interdisciplinary Journal of Nonlinear Science},
	language = {en},
	month = sep,
	number = {3},
	pages = {795--804},
	title = {A coupled ordinary differential equation lattice model for the simulation of epileptic seizures},
	urldate = {2021-05-07},
	volume = {9},
	year = {1999}}

@book{kuramoto_chemical_1984,
	address = {Berlin, Heidelberg},
	author = {Kuramoto, Yoshiki},
	doi = {10.1007/978-3-642-69689-3},
	isbn = {978-3-642-69691-6 978-3-642-69689-3},
	publisher = {Springer Berlin Heidelberg},
	series = {Springer {Series} in {Synergetics}},
	title = {Chemical {Oscillations}, {Waves}, and {Turbulence}},
	urldate = {2021-05-07},
	volume = {19},
	year = {1984}}

@article{hale_diffusive_1997,
	author = {Hale, Jack K.},
	doi = {10.1007/BF02219051},
	issn = {1040-7294, 1572-9222},
	journal = {Journal of Dynamics and Differential Equations},
	language = {en},
	month = jan,
	number = {1},
	pages = {1--52},
	title = {Diffusive coupling, dissipation, and synchronization},
	urldate = {2021-05-07},
	volume = {9},
	year = {1997}}

@article{stankovski_coupling_2017,
	author = {Stankovski, Tomislav and Pereira, Tiago and McClintock, Peter V.E. and Stefanovska, Aneta},
	doi = {10.1103/RevModPhys.89.045001},
	issn = {0034-6861, 1539-0756},
	journal = {Reviews of Modern Physics},
	language = {en},
	month = nov,
	number = {4},
	pages = {045001},
	shorttitle = {Coupling functions},
	title = {Coupling functions: {Universal} insights into dynamical interaction mechanisms},
	urldate = {2021-05-07},
	volume = {89},
	year = {2017}}

@article{witthaut_collective_2022,
	title = {Collective nonlinear dynamics and self-organization in decentralized power grids},
	volume = {94},
	issn = {0034-6861, 1539-0756},
	doi = {10.1103/RevModPhys.94.015005},
	language = {en},
	number = {1},
	urldate = {2022-03-03},
	journal = {Reviews of Modern Physics},
	author = {Witthaut, Dirk and Hellmann, Frank and Kurths, Jürgen and Kettemann, Stefan and Meyer-Ortmanns, Hildegard and Timme, Marc},
	month = feb,
	year = {2022},
	pages = {015005},
}

@article{wolter2018quantifying,
	author = {Wolter, Justine and L{\"u}nsmann, Benedict and Zhang, Xiaozhu and Schr{\"o}der, Malte and Timme, Marc},
	journal = {Chaos: An Interdisciplinary Journal of Nonlinear Science},
	number = {6},
	pages = {063122},
	publisher = {AIP Publishing LLC},
	title = {Quantifying transient spreading dynamics on networks},
	volume = {28},
	year = {2018}}

@article{schroder2019dynamic,
	author = {Schr{\"o}der, Malte and Zhang, Xiaozhu and Wolter, Justine and Timme, Marc},
	journal = {IEEE Transactions on Network Science and Engineering},
	number = {3},
	pages = {1019--1026},
	publisher = {IEEE},
	title = {Dynamic perturbation spreading in networks},
	volume = {7},
	year = {2019}}

@article{tyloo_robustness_2018,
	title = {Robustness of {Synchrony} in {Complex} {Networks} and {Generalized} {Kirchhoff} {Indices}},
	volume = {120},
	issn = {0031-9007, 1079-7114},
	doi = {10.1103/PhysRevLett.120.084101},
	language = {en},
	number = {8},
	urldate = {2022-02-19},
	journal = {Physical Review Letters},
	author = {Tyloo, M. and Coletta, T. and Jacquod, Ph.},
	month = feb,
	year = {2018},
	pages = {084101},
}

@article{plietzsch_bounds_2019,
	title = {Bounds and {Estimates} for the {Response} to {Correlated} {Fluctuations} in {Asymmetric} {Complex} {Networks}},
	copyright = {Creative Commons Attribution 4.0 International},
	doi = {10.48550/ARXIV.1903.09585},
	journal = {arXiv:1903.09585},
	urldate = {2022-03-21},
	author = {Plietzsch, Anton and Auer, Sabine and Kurths, Jürgen and Hellmann, Frank},
	year = {2019},
}

@article{zamora-lopez_cortical_2010,
	title = {Cortical hubs form a module for multisensory integration on top of the hierarchy of cortical networks},
	issn = {16625196},
	doi = {10.3389/neuro.11.001.2010},
	urldate = {2022-03-21},
	journal = {Frontiers in Neuroinformatics},
	author = {{Zamora-López}},
	year = {2010},
}

@article{wang_hierarchical_2019,
	title = {Hierarchical {Connectome} {Modes} and {Critical} {State} {Jointly} {Maximize} {Human} {Brain} {Functional} {Diversity}},
	volume = {123},
	issn = {0031-9007, 1079-7114},
	doi = {10.1103/PhysRevLett.123.038301},
	language = {en},
	number = {3},
	urldate = {2022-03-21},
	journal = {Physical Review Letters},
	author = {Wang, Rong and Lin, Pan and Liu, Mianxin and Wu, Ying and Zhou, Tao and Zhou, Changsong},
	month = jul,
	year = {2019},
	pages = {038301},
}

@article{tononi_measure_1994,
	title = {A measure for brain complexity: relating functional segregation and integration in the nervous system.},
	volume = {91},
	issn = {0027-8424, 1091-6490},
	shorttitle = {A measure for brain complexity},
	doi = {10.1073/pnas.91.11.5033},
	language = {en},
	number = {11},
	urldate = {2022-03-21},
	journal = {Proceedings of the National Academy of Sciences},
	author = {Tononi, G and Sporns, O and Edelman, G M},
	month = may,
	year = {1994},
	pages = {5033--5037},
}

@inproceedings{coletta_transient_2018,
	address = {Miami Beach, FL},
	title = {Transient {Performance} of {Electric} {Power} {Networks} {Under} {Colored} {Noise}},
	isbn = {978-1-5386-1395-5},
	doi = {10.1109/CDC.2018.8619753},
	urldate = {2022-03-21},
	booktitle = {2018 {IEEE} {Conference} on {Decision} and {Control} ({CDC})},
	publisher = {IEEE},
	author = {Coletta, T. and Bamieh, B. and Jacquod, Ph.},
	month = dec,
	year = {2018},
	pages = {6163--6167},
}

\label{lastpage}

\end{document}